\documentclass[conference]{IEEEtran}
\ifCLASSINFOpdf
\else
\fi
\usepackage{etex}
\usepackage{flushend}

\usepackage{amsmath}
\usepackage{amsthm}
\usepackage{amssymb}
\usepackage{xcolor}
\usepackage{tikz}
\usepackage[bookmarks=false]{hyperref}
\usepackage{subcaption}
\usepackage[vlined, ruled, commentsnumbered, linesnumbered]{algorithm2e}
\usepackage{cite}
\usepackage{booktabs}
\usepackage{paralist}
\usepackage{tikz}
\usepackage{multirow}
\usepackage{pgfplots}
\usetikzlibrary{arrows,patterns,automata,positioning,shapes.geometric,plotmarks,patterns,calc}
\usepackage{url} 
\usepackage{pifont}
\usepackage{adjustbox}
\usetikzlibrary{calc}

\pgfdeclareimage[height=16pt, width=16pt]{gears}{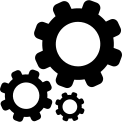}
\pgfdeclareimage[height=8pt, width=12.8pt]{policy}{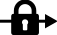}
\pgfdeclareimage[height=12pt, width=19.2pt]{policy-large}{figs/policy.png}
\pgfdeclareimage[height=11pt, width=11pt]{lock}{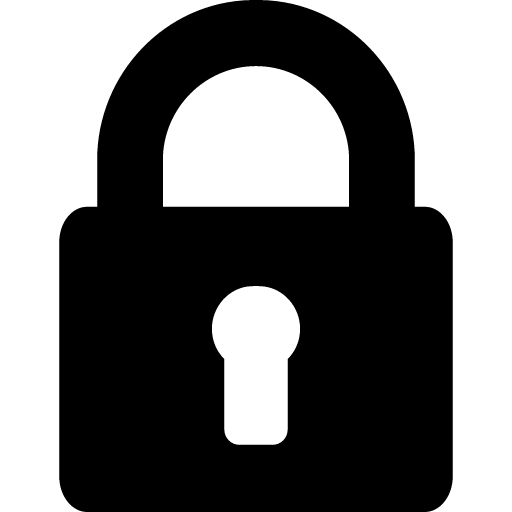}
\pgfdeclareimage[height=20pt, width=20pt]{pep}{figs/lock.png}
\pgfdeclareimage[height=6pt, width=6pt]{lock-small}{figs/lock.png}
\pgfdeclareimage[height=16pt, width=16pt]{shield}{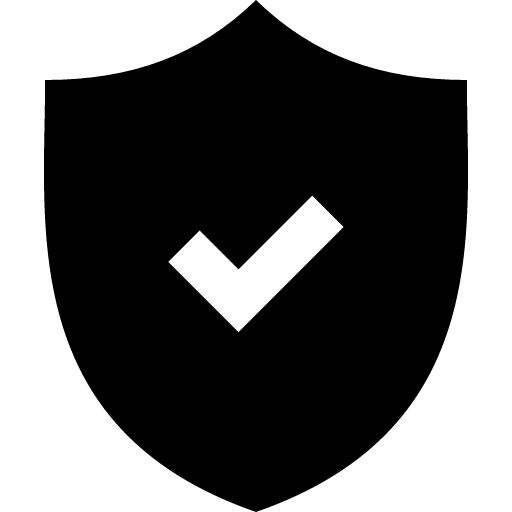}
\pgfdeclareimage[height=10pt, width=10pt]{small-shield}{figs/shield.png}
\pgfdeclareimage[height=20pt, width=20pt]{resource}{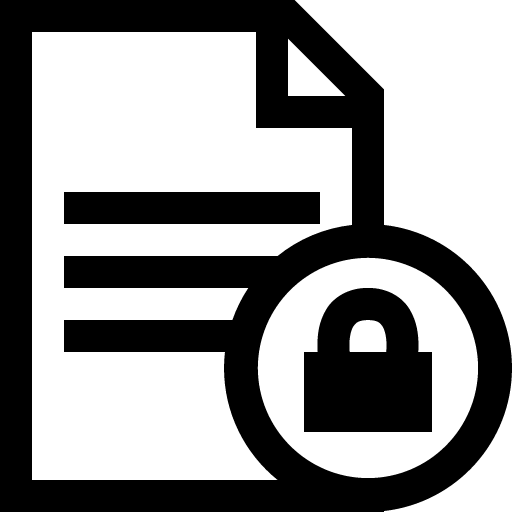}
\pgfdeclareimage[height=14pt, width=14pt]{small-resource}{figs/resource.png}
\pgfdeclareimage[height=20pt, width=20pt]{subject}{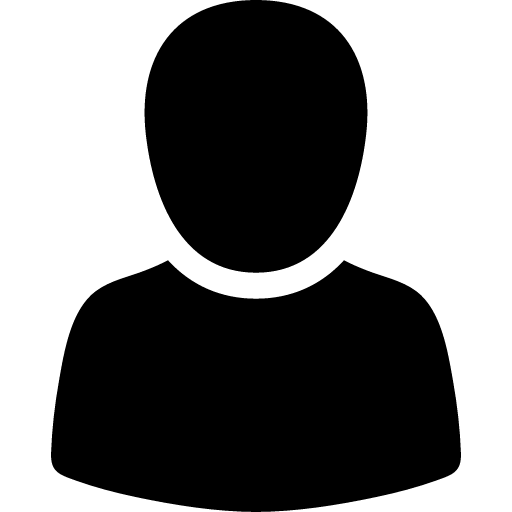}
\pgfdeclareimage[height=20pt, width=20pt]{pip}{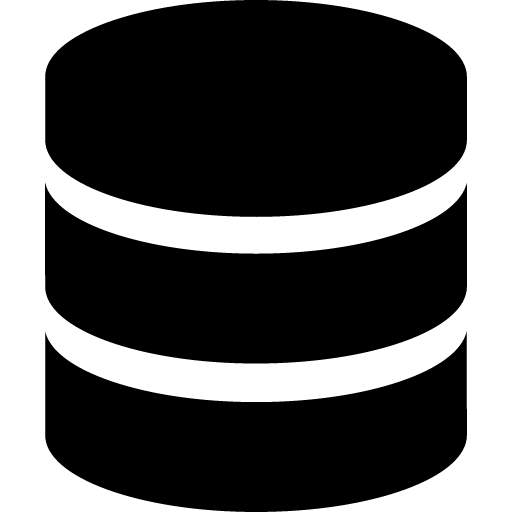}
\pgfdeclareimage[height=10pt, width=20pt]{pdp}{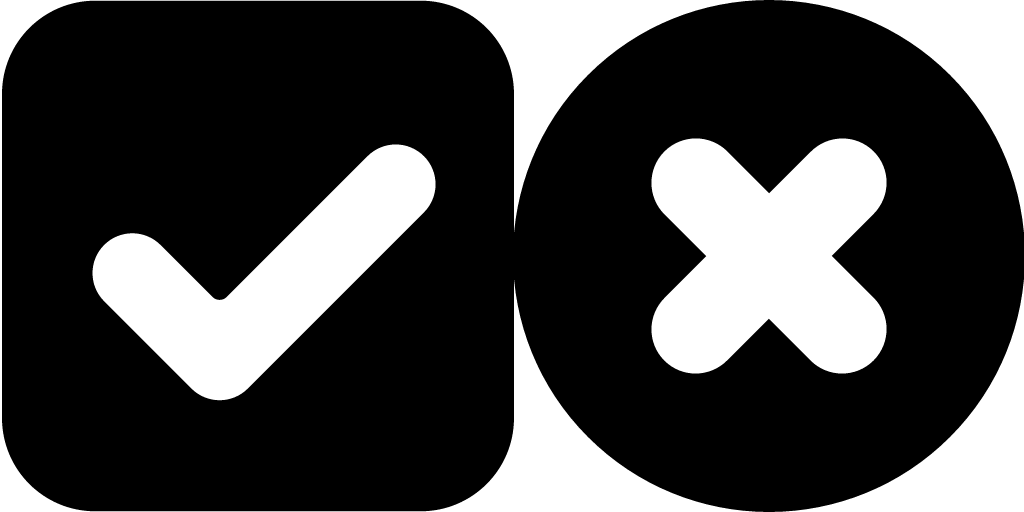}
\pgfdeclareimage[height=16pt, width=16pt]{document}{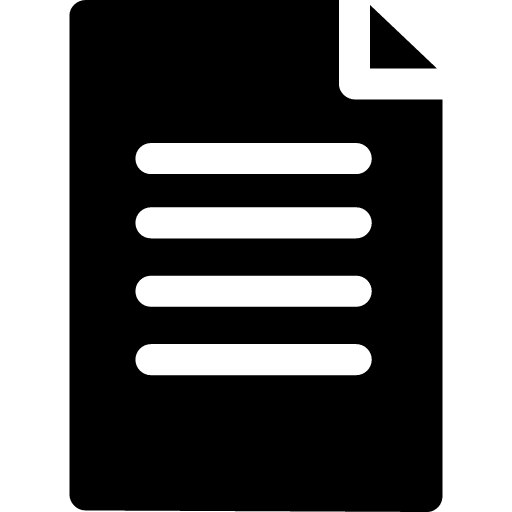}
\pgfdeclareimage[height=20pt, width=20pt]{requirement}{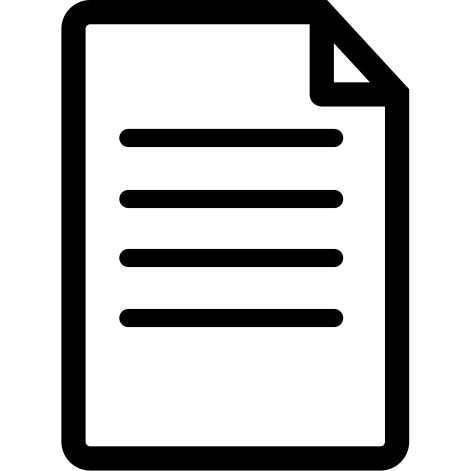}

\begin{document}
%

\newif\ifext
\exttrue

\belowdisplayskip=4pt
\abovedisplayskip=4pt
\setlength{\belowcaptionskip}{-10pt}

\title{Access Control Synthesis for Physical Spaces}

\author{\IEEEauthorblockN{Petar Tsankov\hspace{20pt}Mohammad Torabi Dashti\hspace{20pt}David Basin}
	\IEEEauthorblockA{
		Department of Computer Science\\
		ETH Zurich, Switzerland\\
		\{ptsankov, torabidm, basin\}@inf.ethz.ch}
}


%



\maketitle

\begin{abstract}
%
Access-control requirements for physical spaces, like office buildings
and airports, are best formulated from a global viewpoint in terms of
system-wide requirements.  For example, ``there is an authorized path
to exit the building from every room.''  In contrast, individual
access-control components, such as doors and turnstiles, can only
enforce local policies, specifying when the component may open. In
practice, the gap between the system-wide, global requirements and the
many local policies is bridged manually, which is tedious,
error-prone, and scales poorly.

We propose a framework to automatically synthesize local access-control policies from a set of global requirements for physical
spaces. Our framework consists of an expressive language to specify
both global requirements and physical spaces, and an algorithm for
synthesizing local, attribute-based policies from the global
specification. We empirically demonstrate the framework's
effectiveness on three substantial case studies. The studies
demonstrate that access-control synthesis is practical even for
complex physical spaces, such as airports, with many interrelated security
requirements.








\end{abstract}

%

\newtheorem{definition}{Definition}
\newtheorem{theorem}{Theorem}
\newtheorem{lemma}{Lemma}

\definecolor{darkred}{rgb}{0.75,0,0}
\newcommand{\para}[1]{\vspace{2pt}\noindent{\bf\em#1.\hspace{4pt}}}
\newcommand{\secref}[1]{Section~\ref{#1}}
\newcommand{\remark}[1]{\textcolor{darkred}{\textbf{$\blacktriangleright$#1$\blacktriangleleft$}}}
\newcommand{\sem}[1]{[\![#1]\!]}
\newcommand{\tool}[0]{{\sc PoliSy}\xspace}
\newcommand{\lang}[0]{{\sc SpCtl}\xspace}
\newcommand{\SMT}{{\cal S}_{\sf smt}}
\newcommand{\CS}{{\cal S}_{\sf cs}}
\newcommand{\cmark}{\ding{51}}%
\newcommand{\xmark}{\ding{55}}%
\newcommand{\cl}[1]{{\it cl}(#1)}%

\newcommand{\kaba}{KABA~AG}
\newcommand{\entry}{r_e}

\newcommand{\req}[1]{\hyperref[req:#1]{\bf R#1}\xspace}


\tikzset{entry/.style={
  rectangle,
  fill=black!10, 
  rounded corners,
  line width=1pt,
  draw=black, 
  text=black,
  minimum height=14pt, 
  inner sep=2pt}
}
\tikzset{resource/.style={
  rectangle, 
  rounded corners,
  line width=0.6pt,
  fill=none,
  draw=black,
  text=black,
  minimum height=14pt, 
  inner sep=2pt}
}

\newcommand{\doi}[1]{doi:~\href{http://dx.doi.org/#1}{\Hurl{#1}}}

\setlength{\tabcolsep}{1pt}
\section{Introduction}

\label{sec:intro}

Physical access control is used to restrict access to physical spaces.  For example, it controls who can access which parts of an office building or how personnel can move within critical spaces such as airports or military facilities. As physical spaces are usually comprised of subspaces, such as rooms connected by doors,  policies are enforced by multiple policy enforcement points (PEPs). Each PEP is associated to a control point, like a door, and enforces a \emph{local} policy.

Consider, for example, an office building. An electronic door lock might control access to an office by enforcing a policy that states that only an \emph{employee} may enter the office. This policy is local in the sense that its scope is limited to an individual enforcement point, here the office's door. The policy therefore does not guarantee that non-employees cannot enter the office, since the office may have other doors. Neither does it guarantee that employees can actually access the office. If employees cannot enter the corridor leading to the office's door,  then the local policy is useless.

In contrast to the local policies for enforcement points, 
access-control requirements for physical
spaces are typically \emph{global}.   They express 
constraints on the access paths through the entire space.
In the example above, a requirement might be
that employees should be able to access the office
from the lobby. 
This requirement is global in that no single PEP alone
can guarantee its satisfaction. A standard electronic
lock, which enforces only local policies such as \emph{grant
access to employees}, is oblivious to the physical constraints of the
office building and what policies the other PEPs enforce. It therefore
cannot address this requirement.

\para{Problem Statement} 
The discrepancy between global requirements and local policies creates an abstraction gap that must be bridged when configuring access-control mechanisms.
We consider the problem of automatically synthesizing a set of PEP policies that together enforce global access-control requirements in a given physical space.

This problem is nontrivial.
A given physical space 
usually constrains the ways 
subjects may access its subspaces. These constraints must be accounted for
when configuring the individual PEPs.
Moreover, global access-control requirements may have interdependencies
and hence their individual solutions may not contribute to an overall
solution.
To illustrate this lack of compositionality, suppose in addition to the  requirement that
employees can access an office room from the lobby, we require that
they must not enter the area where auditing documents are stored.
Giving employees access to their office through \emph{any} path
satisfies the first requirement, but it would violate the second one if
the path goes through the audit area.

In practice, constructing local policies for a physical space 
is a manual task where a security engineer writes individual policies, one per PEP,
that collectively enforce the space's global requirements.
This manual process results in errors, such as
granting access to unauthorized subjects or denying access to
authorized ones; the literature contains numerous examples
of such
problems~\cite{Fitzgerald_anomalyanalysis,vmcai11,prediction:codaspy12}.
Moreover, engineers must manually revise their policies whenever 
requirements are changed, or when the physical space changes,
e.g.\ due to construction work.
In short, writing local policies manually is error-prone and
scales poorly.
Our thesis is that it is also unnecessary: the automatic synthesis of
local policies with system-wide security guarantees is a viable
alternative.

\begin{figure}[t]
\centering
\begin{tikzpicture}[->,>=stealth,shorten >=1pt,auto,line width=0.6pt]
\footnotesize
  \node[anchor=north east, rectangle, draw=black, minimum width=60pt, minimum height=45pt, line width=1pt] (main) at (0,0) {};
  \node[anchor=north east, rectangle, draw=black, minimum width=25pt, minimum height=15pt, line width=1pt] (bur) at (0,0) {};
  \node[anchor=north east, rectangle, draw=black, minimum width=40pt, minimum height=30pt, line width=1pt] (cor) at (0,0) {};
  \node[anchor=south east, rectangle, draw=black, minimum width=40pt, minimum height=15pt, line width=1pt] (mr) at (main.south east) {};

  \node[rectangle, rounded corners, fill=red!20, draw=black, minimum width=11pt, minimum height=11pt, inner sep=0pt] at (bur.west) {\bf ?};
  \node[rectangle, rounded corners, fill=red!20, draw=black, minimum width=11pt, minimum height=11pt, inner sep=0pt] at (cor.west) {\bf ?};
  \node[rectangle, rounded corners, fill=red!20, draw=black, minimum width=11pt, minimum height=11pt, inner sep=0pt] at (main.west) {\bf ?};
  \node[rectangle, rounded corners, fill=red!20, draw=black, minimum width=11pt, minimum height=11pt, inner sep=0pt] at (main.east) {\bf ?};
  \node[rectangle, rounded corners, fill=red!20, draw=black, minimum width=11pt, minimum height=11pt, inner sep=0pt] at (mr.north) {\bf ?};
  \node[anchor=north] at (main.south) {Physical space};

  \node[anchor=north] (req1) at ($(main) - (0,1.4)$) {\pgfuseimage{requirement}};
  \node[anchor=north west] (req2) at ($(req1.north west) + (0.1,-0.08)$) {\pgfuseimage{requirement}};
  \node[anchor=north] at (req2.south) {Global requirements};
  \node[anchor=north west] (req3) at ($(req2.north west) + (0.1,-0.08)$) {\pgfuseimage{requirement}};

  \node[rectangle, rounded corners, draw=black, line width=1pt, inner sep=4pt, anchor=west] (synth) at ($(main.south) + (2,-0.2)$) {\begin{tabular}{c}\pgfuseimage{gears}\\\bf \small Synthesizer\end{tabular}};

  \node[anchor=west, rectangle, draw=black, minimum width=60pt, minimum height=45pt, line width=1pt] (main) at ($(synth.east) + (1,0)$) {};
  \node[anchor=north east, rectangle, draw=black, minimum width=25pt, minimum height=15pt, line width=1pt] (bur) at (main.north east) {};
  \node[anchor=north east, rectangle, draw=black, minimum width=40pt, minimum height=30pt, line width=1pt] (cor) at (main.north east) {};
  \node[anchor=south east, rectangle, draw=black, minimum width=40pt, minimum height=15pt, line width=1pt] (mr) at (main.south east) {};

  \node[rectangle, rounded corners, fill=green!20, draw=black, minimum width=11pt, minimum height=11pt, inner sep=0pt] at (bur.west) {\pgfuseimage{lock-small}};
  \node[rectangle, rounded corners, fill=green!20, draw=black, minimum width=11pt, minimum height=11pt, inner sep=0pt] at (cor.west) {\pgfuseimage{lock-small}};
  \node[rectangle, rounded corners, fill=green!20, draw=black, minimum width=11pt, minimum height=11pt, inner sep=0pt] at (main.west) {\pgfuseimage{lock-small}};
  \node[rectangle, rounded corners, fill=green!20, draw=black, minimum width=11pt, minimum height=11pt, inner sep=0pt] at (main.east) {\pgfuseimage{lock-small}};
  \node[rectangle, rounded corners, fill=green!20, draw=black, minimum width=11pt, minimum height=11pt, inner sep=0pt] at (mr.north) {\pgfuseimage{lock-small}};
  \node[anchor=north] (l1) at (main.south) {Physical space};
  \node[anchor=north] at ($(l1.south) + (0,0.1)$) {with local policies};

  \draw[draw=black, line width=1pt] ($(synth.north west) - (0.8,1.3)$)--($(synth.west) + (-0.1, -0.1)$);
  \draw[draw=black, line width=1pt] ($(synth.south west) - (0.8,-1.3)$)--($(synth.west) + (-0.1, 0.1)$);
  \draw[draw=black, line width=1pt] ($(synth.east) + (0.1,0)$)--($(synth.east) + (0.7, 0)$);
\end{tikzpicture}
\caption{Access control synthesis for physical spaces}
\label{fig:intro}
\end{figure}
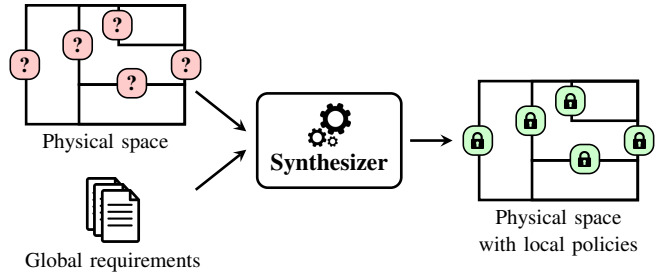

\para{Approach and Contributions} We propose a formal framework for automatically
synthesizing local policies that run on distributed PEPs
from a set of global access-control
requirements for a given physical space.
The framework's main ingredients are depicted in
Figure~\ref{fig:intro}.
The key component is a \emph{synthesizer}, which takes as input a
model of the physical space and a set of global requirements.
The synthesizer's output is the set of local policies that the PEPs enforce.
If the global requirements are satisfiable, then the synthesizer is guaranteed to output a correct set of local policies; otherwise, it returns ${\sf unsat}$ to indicate that the requirements cannot be satisfied.
Hence, using our framework, engineers can generate local policies from global requirements simply by formalizing the global requirements and modeling the physical space.

Below, we briefly describe the framework's components, depicted in
Figure~\ref{fig:intro}.  We use directed graphs to model physical
spaces: a node represents an enclosed space, such as an office or a
corridor, and an edge represents a PEP, for example installed on a
door or turnstile. The nodes are labeled to denote their
attributes. These attributes may include the assets the node contains
(audit documents), its physical attributes (international terminal),
and its clearance level (high security zone). These attributes may be
used when specifying policies.
Formally, our model of a physical space is a Kripke structure.

We give a declarative language, called \lang, for specifying global requirements.
Our language is built on the computation tree logic
(CTL)~\cite{Emerson1982241} and supports subject attributes (e.g., an organizational
role), time constraints (e.g., business-hour requirements), as well
as quantification over paths and branches in physical spaces.
To demonstrate its expressiveness, we show how 
common physical access-control requirements can be directly written in \lang.
Moreover, to simplify the task of formalizing such requirements, we
develop requirement patterns and illustrate their use through
examples.
Our synthesis algorithm outputs 
attribute-based policies, 
expressed as constraints over subject attributes and contextual
conditions, such as organizational roles and the current time. This
covers a wide range of practical setups and scenarios, including
attribute-based and
role-based access control.
We strike a balance between the requirement language's expressiveness
and the complexity of synthesizing local policies.
The synthesis problem we consider is NP-hard.
However, we show that
for practically-relevant requirements, it can 
be efficiently solved using existing SMT solvers.
This is intuitively because physical spaces, in practice,
induce directed graphs that have 
short simple-paths.  We illustrate our framework's effectiveness
 using three case studies where we synthesize
access-control policies for a university building, a corporate
building, and an airport terminal.   Synthesizing local
policies in each case takes less than $30$ seconds.
The last two case studies are based on real-world examples developed
together with \kaba, a leading physical access control company.

To the best of our knowledge, ours is the first framework for
synthesizing policies from system-wide access-control requirements.
We thereby solve a fundamental problem in access control for physical spaces.
An immediate practical consequence is that security engineers can
focus on 
system-wide requirements, and delegate to our synthesizer the task of
constructing the local policies with correctness guarantees.
We remark that although this work is focused on access control for
physical spaces, the ideas presented are general and can be extended
to other domains, such as computer networks partitioned into subnetworks
by distributed firewalls. 

\para{Organization} 
We give an overview of our access-control synthesis framework
in~\secref{sec:overview}. In~\secref{sec:phys-spec}, we describe and
formalize our system model. In~\secref{sec:reqs}, we define our
\lang language for specifying global requirements, and
present requirement patterns. 
In \secref{sec:synthesis}, we 
define the policy synthesis problem and prove its decidability. In \secref{sec:algorithm}, we define an efficient policy synthesis algorithm. In \secref{sec:eval}, we describe our implementation and
report on our experiments. We review related work in \secref{sec:rw},
and we draw conclusions and discuss future work in~\secref{sec:conc}. 
\ifext
The appendices contain all proofs.
\else
Proofs can be found in the extended version of this paper~\cite{tech-report}.
\fi

\section{Overview} \label{sec:overview} 

We start with a simple example that illustrates the challenges
of constructing local policies that cumulatively enforce global
access-control requirements. We also explain how our framework is
used, that is, we describe its inputs and outputs.

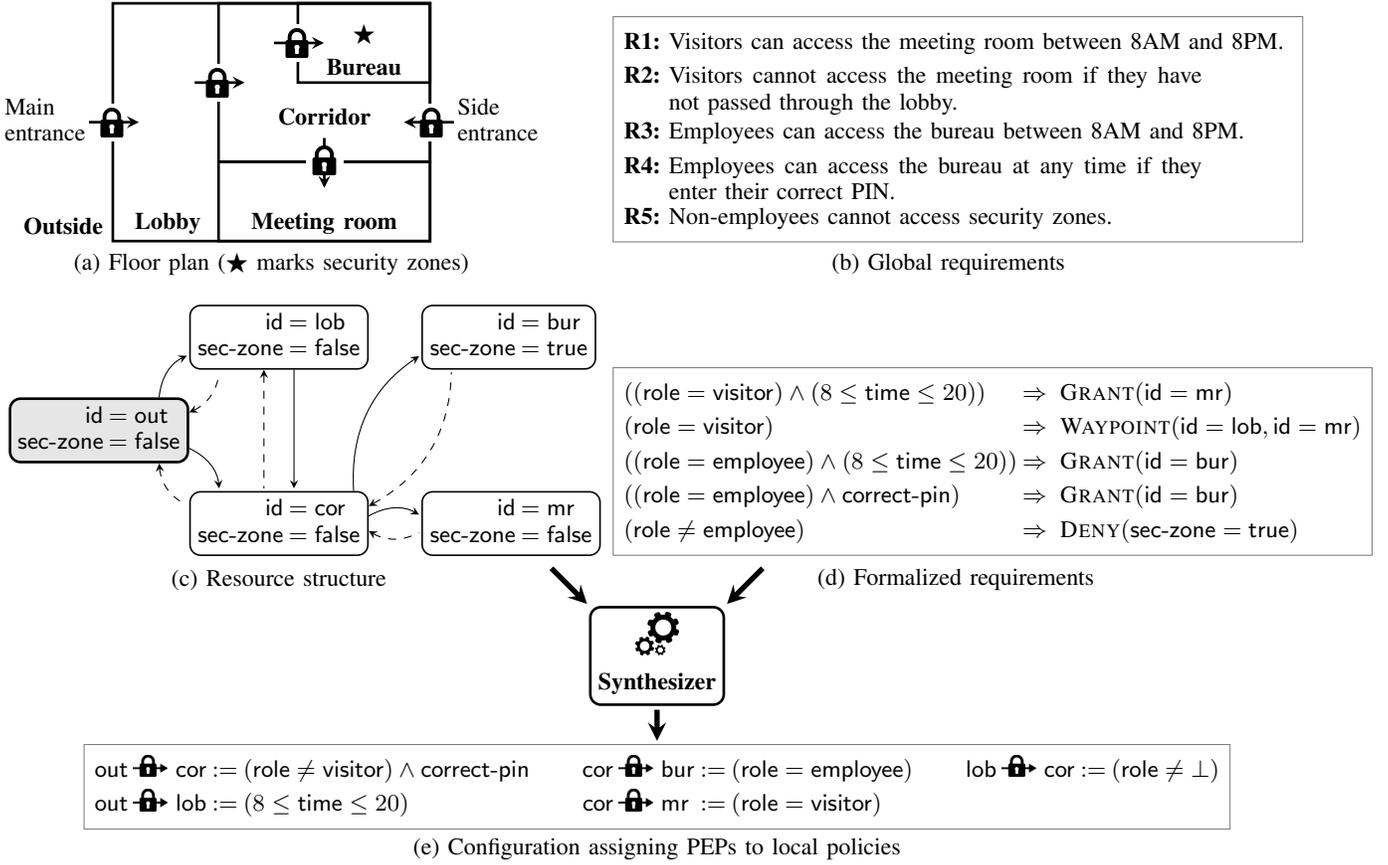
\begin{figure*}[t]
\centering
\hspace{-10pt}
\small
\begin{tikzpicture}[->,>=stealth,shorten >=1pt,auto]

  \node[anchor=north east, rectangle, draw=black, minimum width=120pt, minimum height=90pt, line width=1pt] (main) at (0,0) {};
  \node[anchor=south east] at (main.south west) {\bf Outside};
  \node[anchor=south west] at ($(main.south west) + (0.2,0)$) {\bf Lobby};
  \node[anchor=north east, rectangle, draw=black, minimum width=50pt, minimum height=30pt, line width=1pt] (bur) at (0,0) {};
  \node[anchor=south] at (bur.south) {\bf Bureau};
  \node at ($(bur) + (0, 0.1)$) {$\bigstar$};
  \node[anchor=north east, rectangle, draw=black, minimum width=80pt, minimum height=60pt, line width=1pt] (cor) at (0,0) {};
  \node[anchor=south] at ($(cor.south) + (0,0.4)$) {\bf Corridor};
  \node[anchor=south east, rectangle, draw=black, minimum width=80pt, minimum height=30pt, line width=1pt] (mr) at (main.south east) {};
  \node[anchor=south] at (mr.south) {\bf Meeting room};

  \node[fill=white, inner sep=4pt] (bur-white) at (bur.west) {};
  
  \node[inner sep=7pt, fill=white] at (main.west) {};
  \draw[line width=1pt, >=stealth, ->] ($(main.west) - (0.3,0)$)--($(main.west) + (0.4,0)$);
  \node[inner sep=0pt] (main-lock) at (main.west) {\pgfuseimage{lock}};

  \node[inner sep=7pt, fill=white] at (bur.west) {};
  \draw[line width=1pt, >=stealth, ->] ($(bur.west) - (0.3,0)$)--($(bur.west) + (0.4,0)$);
  \node[inner sep=0pt] (bur-lock) at (bur.west) {\pgfuseimage{lock}};

  \node[inner sep=7pt, fill=white] at (cor.west) {};
  \draw[line width=1pt, >=stealth, ->] ($(cor.west) - (0.3,0)$)--($(cor.west) + (0.4,0)$);
  \node[inner sep=0pt] (cor-lock) at (cor.west) {\pgfuseimage{lock}};

  \node[inner sep=7pt, fill=white] at (main.east) {};
  \draw[line width=1pt, >=stealth, ->] ($(main.east) + (0.3,0)$)--($(main.east) - (0.4,0)$);
  \node[inner sep=0pt] (side-lock) at (main.east) {\pgfuseimage{lock}};

  \draw[line width=1pt, >=stealth, ->] ($(mr.north) + (0,0.3)$)--($(mr.north) - (0,0.4)$);
  \node[fill=white, minimum width=14pt, minimum height=10pt] (mr-lock-holder) at (mr.north) {};
  \node[inner sep=0pt] (mr-lock) at (mr.north) {\pgfuseimage{lock}};

  \node[anchor=east] (main-entrance) at ($(main-lock) - (0.2,0)$) {\begin{tabular}{l}Main\\entrance\end{tabular}};
  \node[anchor=west] at ($(side-lock) + (0.2,0)$) {\begin{tabular}{l}Side\\entrance\end{tabular}};

  \node[minimum width=12em] (l) [below=0.1em of main.south] {(a) Floor plan ($\bigstar$ marks security zones)};

  \node[anchor=west, entry] (out)     at ($(main-entrance.west) +
(0.2,-4.1)$) {\begin{tabular}{rl}$\sf id$ & $\sf = out$\\$\sf \operatorname{\sf sec-zone}$ & $\sf =false$\end{tabular}};
  \node[resource, anchor=west] (lob)  at ($(out.north east) + (0, 0.8)$) {\begin{tabular}{rl}$\sf id$ & $\sf = lob$\\$\sf \operatorname{\sf sec-zone}$ & $\sf =false$\end{tabular}};
  \node[resource, anchor=west] (cor)  at ($(out.south east) - (0, 0.8)$) {\begin{tabular}{rl}$\sf id$ & $\sf = cor$\\$\sf \operatorname{\sf sec-zone}$ & $\sf =false$\end{tabular}};
  \node[resource, anchor=west] (mr)  at ($(cor.east) + (0.7, 0)$) {\begin{tabular}{rl}$\sf id$ & $\sf = mr$\\$\sf \operatorname{\sf sec-zone}$ & $\sf =false$\end{tabular}};
  \node[resource, anchor=west] (bur) at ($(lob.east) + (0.7, 0)$) {\begin{tabular}{rl}$\sf id$ & $\sf = bur$\\$\sf \operatorname{\sf sec-zone}$ & $\sf =true$\end{tabular}};
  \path 
        ($(cor.east) + (0, 0.1)$) edge[->, bend left] ($(mr.west) + (0, 0.1)$)
        ($(cor.east) - (0, 0.1)$) edge[<-, bend right, dashed] ($(mr.west) - (0, 0.1)$)
        ($(cor.north east) - (0.2, 0)$) edge[->, bend left] ($(bur.south west) + (0, 0.2)$)
        ($(cor.north east) - (0, 0.2)$) edge[<-, bend right, dashed] ($(bur.south west) + (0.4, 0)$)
        ($(out.south east) + (0, 0.2)$) edge[->, bend left] ($(cor.north west) + (0.4, 0)$)
        ($(out.south east) - (0.4, 0)$) edge[<-, bend right, dashed] ($(cor.north west) - (0, 0.2)$)
        ($(out.north east) - (0.4, 0)$) edge[->, bend left] ($(lob.south west) + (0, 0.2)$)
        ($(out.north east) - (0, 0.2)$) edge[<-, bend right, dashed] ($(lob.south west) + (0.4, 0)$)
        ($(lob.south) - (0.2,0)$) edge[<-, dashed] ($(cor.north) - (0.2, 0)$)
        ($(lob.south) + (0.2,0)$) edge[->] ($(cor.north) + (0.2, 0)$);
  \node[minimum width=12em] (constr) [below=0.2em of cor] {(c) Resource structure};

\renewcommand{\arraystretch}{1.3}

\node[anchor=south, draw=gray] (Reqs) at ($(main.south east) + (7,0)$)    
  {\begin{tabular}{l}
    {\bf R1:} Visitors can access the meeting room between 8AM and 8PM. \label{req:1}\\
    {\bf R2:} \parbox[t]{0.8\columnwidth}{Visitors cannot access the meeting room if they have not passed through the lobby.} \label{req:2}\\
    {\bf R3:} Employees can access the bureau between 8AM and 8PM. \label{req:3}\\
    {\bf R4:} \parbox[t]{0.8\columnwidth}{Employees can access the bureau
    at any time if they enter their correct PIN.} \label{req:4}\\
    {\bf R5:} Non-employees cannot access security zones. \label{req:5}
  \end{tabular}
  };

  \node[minimum width=12em] at ($(l) + (9,0)$) {(b) Global requirements};


  \node[anchor=west, draw=gray] (Reqs) at ($(Reqs.west) - (0,4.42)$)
    {\begin{tabular}{lll}
      $(({\sf role} = {\sf visitor}) \wedge (8\leq {\sf time}\leq 20))$ & $\Rightarrow \textsc{ Grant}({\sf id} = {\sf mr})$\\
      $({\sf role} = {\sf visitor})$ & $\Rightarrow \textsc{ Waypoint}({\sf id} = {\sf lob}, {\sf id} = {\sf mr})$\\
      $(({\sf role} = {\sf employee}) \wedge (8\leq {\sf time}\leq 20))$ & $\Rightarrow \textsc{ Grant}({\sf id} = {\sf bur})$\\
      $(({\sf role} = {\sf employee}) \wedge \operatorname{\sf correct-pin})$ & $\Rightarrow  \textsc{ Grant}({\sf id} = {\sf bur})$\\
      $({\sf role} \neq {\sf employee})$ & $\Rightarrow \textsc{ Deny}(\operatorname{\sf sec-zone}={\sf true})$
    \end{tabular}};
  \node[minimum width=31em] (reqs-label) at ($(constr) + (9,0)$)  {(d) Formalized requirements};


  \node[rectangle, rounded corners, draw=black, line width=1pt, inner sep=2pt] (synth) at (3,-8.7) {\begin{tabular}{c}\pgfuseimage{gears}\\\bf Synthesizer\end{tabular}};
  \node (in1) [above left=2em of synth] {};
  \node (in2) [below left=2em of synth] {};
  \node (out1) [right=2em of synth] {};
  \draw[->, draw=black, line width=2pt] ($(synth.north west) + (-0.5, 0.5)$)--($(synth.north west) + (0,0)$);
  \draw[->, draw=black, line width=2pt] ($(synth.north east) + (0.5, 0.5)$)--($(synth.north east) + (0,0)$);
  \draw[->, draw=black, line width=2pt] ($(synth.south) + (0,
-0.05)$)--($(synth.south) + (0,-0.5)$);


  \node[anchor=north, draw=gray] (policies) [below=1.6em of synth] 
    {\begin{tabular}{llllllllllll}
${\sf out}$ & $\pgfuseimage{policy}\ {\sf cor}$  & $:=$  & $({\sf role} \neq {\sf visitor}) \wedge \operatorname{\sf correct-pin}\qquad$ 
& ${\sf cor}$ & $\pgfuseimage{policy}\ {\sf bur}$  & $:=$  & $({\sf role} = {\sf employee})\qquad$
& ${\sf lob}$ & $\pgfuseimage{policy}\ {\sf cor}$  & $:=$  & $({\sf role}\neq \bot)$ \\
${\sf out}$ & $\pgfuseimage{policy}\ {\sf lob}$  & $:=$  & $(8\leq {\sf time}\leq 20)$
& ${\sf cor}$ & $\pgfuseimage{policy}\ {\sf mr}$  & $:=$  &  $({\sf role} = {\sf visitor})$\\
    \end{tabular}};
  \node[minimum width=14em] (pol-lab) [below=0em of policies] {(e) Configuration assigning PEPs to local policies};
\end{tikzpicture}
\caption{Synthesizing the local policies for our running example}
\label{fig:synth-overview}
\end{figure*}

\subsection{Running Example} \label{sec:example} Consider a small
office space consisting of a lobby, a bureau, a meeting room, and a
corridor. The office layout is given in
Figure~\ref{fig:synth-overview}(a).
Access within this physical space is secured using electronic locks.
Each door has a lock and a card reader. The lock stores a policy
that defines who can open the door from the card reader's side. The
door can be opened by anyone from the opposite side. We annotate
locks with arrows
in Figure~\ref{fig:synth-overview}(a) to indicate the direction that the locks
restrict access. For example, the lock at the main entrance restricts who can
access the lobby, and it allows anyone to exit the office space from
the lobby.
To open a door from the card reader's side, a subject presents a
smartcard that stores the holder's credentials.  The lock can access
additional information, such as the current time, needed to evaluate the
policy. The lock opens whenever the policy evaluates to grant. 

The global requirements for this physical space are given in
Figure~\ref{fig:synth-overview}(b). The requirements \textbf{R1},
\textbf{R3}, and {\bf R4} define permissions, while \textbf{R2} and
\textbf{R5} define prohibitions.  To meet these requirements, the
electronic locks must be configured with appropriate local policies.
As previously observed, this is challenging because
one must account for both spatial constraints 
and all global access-control requirements.
We illustrate these points below.

\para{Spatial Constraints} The layout of the physical space
prevents subjects from
freely requesting access to any resource. 
For example, the requirement \textbf{R1} is not met just
because the meeting room's lock grants access to visitors;
the visitor must also be able to enter the corridor from the outside.
Such constraints must be accounted for when defining
the local policies.
To satisfy \textbf{R1}, we may for instance choose a path from the main
entrance to the meeting room and configure all the locks along that
path to grant access to visitors.

\para{Global Requirements} 
Each global requirement typically has
multiple sets of local policies that satisfy it. The
local policies must however be constructed to ensure that
\emph{all} requirements are satisfied \emph{simultaneously}. For example, the
requirement \textbf{R1} is satisfied if the side-entrance lock and the
meeting room lock both grant access to visitors between $8$AM and $8$PM.  It can also be
satisfied by ensuring that the main entrance, the lobby, and the
meeting room locks all grant access to visitors between $8$AM and $8$PM.  Granting visitors
access through the side entrance however violates the requirement
\textbf{R2}, which requires that visitors pass through the lobby.
Hence, to meet both requirements, the locks along the path through
the lobby must grant access to visitors between $8$AM and $8$PM, while the side-entrance lock
must always deny access to visitors.

\subsection{Synthesis Framework}
Figure~\ref{fig:synth-overview}(c-e) depicts our framework's input and output for our running example. The input is a model of the physical space and a specification of its global requirements. The output produced by our synthesizer is a set of local policies.


A physical space is modeled as a rooted directed graph 
called a \emph{resource structure}.  We have depicted
the root node in gray.  In our example, this
corresponds to the public space that surrounds the office space, e.g.\
public streets. The remaining (non-root) nodes are the spaces 
inside the building. The locks control access along the edges.
A subject can traverse a solid edge of the resource structure
only if the lock's policy evaluates to
grant, whereas any subject can follow the dashed edges. Hence, the locks
effectively enforce the grant-all policy along the dashed edges.
We use two attributes to label the physical spaces: the
attribute~$\mathsf{id}$ represents room identifiers, and $\operatorname{\sf sec-zone}$ formalizes that a space is inside the security zone. 

Global requirements are specified using a declarative language, called
\lang. In Figure~\ref{fig:synth-overview}(d) we show the formalization
of our running example's requirements in \lang.
For instance,~{\bf R1}, which states that visitors can access the
meeting room between $8$AM and $8$PM, is formalized as $\big(\mathsf{(role
= visitor)} \wedge \mathsf{(8 \le time\le 20)}\big)\Rightarrow
\textsc{Grant}\mathsf{(id = mr)}$.
This formalization instantiates \lang's permission pattern
$\textsc{Grant}$ to state that there is a path from outside to
the meeting room such that every lock on the path grants access to any
visitor between $8$AM and $8$PM.
We define \lang's syntax and semantics and present several patterns in~\secref{sec:reqs}.

Given these inputs, the synthesizer automatically constructs a local policy for each lock.
The synthesized policies are attribute-based policies that
collectively enforce the global requirements. The synthesized policies
for our running example are given in
Figure~\ref{fig:synth-overview}(e).
We write, for example, ${\sf cor}\ \pgfuseimage{policy}\ {\sf bur}$
for the synthesized policy deployed at the bureau's lock. This
policy $\sf (role = employee)$
grants access to subjects with the role $\sf employee$. We define the synthesis problem, and the syntax
and semantics of attribute-based local policies
in~\secref{sec:synthesis}.


\section{Physical Access Control}
\label{sec:phys-spec}

\subsection{Basic notions}

In this section, we formalize our system model for physical access control.   Our terminology is based, in part, on the XACML reference architecture \cite{xacml3}.

Each physical space is partitioned into finitely many enclosed spaces
and one open (public) space. We call the enclosed spaces
\emph{resources}.  Two spaces may be directly connected with a
\emph{gate}, controlled by a \emph{policy enforcement point} (PEP).
Examples of gates include doors, turnstiles, and security checkpoints.
Each PEP has its own \emph{policy decision point} (PDP), which stores
a \emph{local policy} mapping access requests to access decisions.
Each {\em access request} consists of subject credentials, which the
PDP receives from the PEP, as well as contextual attributes, if needed,
obtained from \emph{policy information points} (PIPs).  The PIPs are
distributed information sources that provide contextual attributes
required by the PDP for access decisions.  Examples of PIPs
include revocation list servers and secure time servers.

To enter a space, a subject provides his credentials to the PEP that
controls the gate.  The PEP forwards the subject's credentials
to its PDP.  The PDP, in turn, queries the PIP if needed, evaluates the
policy, and then forwards the access decision --- either grant or deny ---
to the PEP.   The PEP then enforces the PDP's
decision. See Figure~\ref{fig:sys-model}.

We assume that access requests contain all relevant information for making access decisions.
PDPs can thus make their decisions independent of past access requests. Hence it has no bearing on our model whether the PDPs are actually distributed or are realized through a centralized system. 
This is desirable from a practical standpoint since PDPs and PIPs need not be equipped with logging mechanisms. Moreover, different PDPs and PIPs need not synchronize their local views on the request history. In this sense they are autonomous entities.


The access requests and local policies we consider are attribute
based and may reference three kinds of attributes.
A \emph{subject attribute} contains information about a subject. For
example, Alice's \emph{organizational role} and \emph{clearance level}
are her subject attributes. Subjects can provide PEPs with their
attributes in the form of credentials.
A {\em contextual attribute} represents information about the security
context provided by a PIP, such as the list of revoked credentials and the current time.
We also introduce {\em resource attributes}, which represent
information about resources.  For example, the attributes \emph{floor}
and \emph{department} may represent the floor of an office space and
the department it belongs to.  We use resource attributes to specify
global requirements. They are however not needed for expressing access
requests or local policies in our model.
This is because the PDPs associated to any resource can be hardwired
with all the attributes of that resource. In this sense, each PDP 
``knows'' the space under its control.

Our system model targets electronic PEP/PDPs 
that can 
enforce attribute-based policies and read digital credentials,
e.g.\ stored on smart cards and mobile phones.
Manufacturers often refer to these as smart locks~\cite{goji,augustus,lockitron}.
In large physical access-control systems, smart locks are
rapidly replacing mechanical locks and keys, which can only enforce 
simple, crude policies.

\begin{figure}
\begin{tikzpicture}[->,>=stealth,shorten >=1pt,auto]
\small
  \node (subj) at (0,0) {\pgfuseimage{subject}};
  \node[anchor=north] (subj-label) at (subj.south) {Subject};

  \node (pep) [right=1.7 of subj] {\pgfuseimage{pep}};
  \node[anchor=north] (pep-label) at (pep.south) {PEP};
  \draw[->] ($(subj.east) + (0,0.2)$) -- node [above=-1pt] {Credentials} ($(pep.west) + (0,0.2)$);
  \draw[<-] ($(subj.east) + (0,-0.3)$) -- node [above=-2pt] {Grant/Deny} ($(pep.west) + (0,-0.3)$);

  \node[draw=black, rounded corners, minimum height=24pt] (pdp) [right=1.7 of pep] {\pgfuseimage{pdp}};
  \node[anchor=north] (pdp-label) at (pdp.south) {PDP};
  \draw[->] ($(pep.east) + (0,0.2)$) -- node [above=-1pt] {Credentials} ($(pdp.west) + (0,0.2)$);
  \draw[<-] ($(pep.east) + (0,-0.3)$) -- node [above=-2pt] {Grant/Deny} ($(pdp.west) + (0,-0.3)$);

  \node[anchor=west] (policy) at ($(pdp) + (0.8,1)$) {\pgfuseimage{document}};
  \node[anchor=south west] (policy-label) at ($(policy.south east) - (0.2,0)$) {Policy};
  \draw[-, dashed] ($(pdp.north)$) -- ($(policy.west)$);

  \node (pip) [right=1.7 of pdp] {\pgfuseimage{pip}};
  \node[anchor=north] (pip-label) at (pip.south) {PIP};
  \draw[<-] ($(pdp.east) + (0,0)$) -- node [above=-2pt] {Attributes} ($(pip.west) + (0,0)$);
\end{tikzpicture}
\caption{System model}
\label{fig:sys-model}
\end{figure}
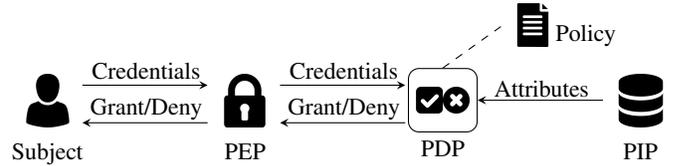

\subsection{Formalization}

We now formalize the above notions.

\para{Attributes}
Fix a finite set~$\cal A$ of \emph{attributes}
and a set~$\cal V$ of attribute {\em values}.
The {\em domain} function ${\sf dom}\colon {\cal A}\to {\cal P}({\cal
  V})$ associates each attribute with the set of values it admits.
For instance, the current time attribute is associated with the set of
natural numbers, and the clearance level attribute is associated with
a fixed finite set of levels.
We assume that any attribute can take the designated
value~$\bot$, representing the situation where the
attribute's value is unknown.
%
We partition the set of attributes into subject attributes~${\cal
  A}_S$, contextual attributes~${\cal A}_C$, and resource
attributes~${\cal A}_R$.  

\para{Access Requests} We represent an access request as a total function that
maps subject and contextual attributes to values from their respective
domains. 
This function is computed by PDPs after receiving a subject's
credentials and querying PIPs.
For instance, the PDP maps the attribute~$\mathsf{role}$
to~$\mathsf{visitor}$ when the subject's credentials indicate
this. It maps the attribute~$\operatorname{\sf correct-pin}$ to
$\mathsf{true}$ when the PIN entered through the keypad
attached to the PDP is correct. Finally, it maps the contextual
attribute~$\mathsf{time}$ to~$8$ after querying a time server at $8$AM.
We denote the set of all access requests by~${\cal Q}$.

A remark on set-valued attributes is due here. In some settings, attributes take a finite set of values, as opposed to a single value. 
For example, in role-based access control, a subject may activate multiple roles.
The attribute $\sf role$ must then be assigned with the set of all the activated roles.
We account for such set-valued attributes simply by defining a Boolean attribute for each value; for example, we define $\sf role\_employee$ and $\sf role\_manager$. An access request~$q$ 
assigns true to both Boolean attributes whenever a subject has activated both the employee and the manager roles.

\para{Local Policies} 
Local policies map access requests to grant or deny. We extensionally define local policies as subsets of~$\mathcal{Q}$: a local policy is defined as the set of requests that it grants. 
The structure $({\cal P}({\cal Q}), \subseteq, \cap, \cup, \emptyset, {\cal Q})$ is a complete lattice that orders local policies by their permissiveness. The least permissive policy, namely~$\emptyset$, denies all access requests, and the most permissive one, i.e.~$\mathcal{Q}$, grants them all.
In section~\ref{sec:problem}, we will intensionally define local policies as constraints over subject and contextual attributes. 
The local policies shown in Figure~\ref{fig:synth-overview}, for example, are defined by such constraints.


\para{Resource Structures}
We now give a formal model of physical spaces.
A \emph{resource structure} is a tuple~$S = ({\cal R}, E, \entry, L)$, where ${\cal R}$ is a set of {\em resources}, $E\subseteq {\cal R}\times {\cal R}$ is an irreflexive edge relation, $\entry\in {\cal R}$ is the \emph{entry resource}, and $L: {\cal R}\to ({\cal A}_R\to {\cal V})$ is a total function mapping resources to resource attribute valuations.
We assume that every resource~$r\in \mathcal{R}$ is reachable from the entry resource~$\entry$, that is,~$(\entry,r)\in E^*$, where~$E^*$ is the reflexive-transitive closure of~$E$.

The edges in a resource structure model PEPs. The irreflexivity of
$E$ captures the condition that once a subject enters a physical
space, he cannot re-enter the space before first leaving it.
We assume that resource structures do not contain \emph{deadlocks}. A
resource~$r_0$ in~$S$ is a deadlock if there does not exists an~$r_1$
such that~$(r_0,r_1)\in E$. This assumption is valid in physical-space
access control: a deadlock resource corresponds to a ``black hole''
that no one can leave. Note that dead-end corridors are not deadlocks,
provided one can backtrack.

The entry resource~$\entry$ represents the public space and the remaining
resources denote enclosed spaces.
A resource structure describes how subjects can access resources.
A subject accesses a resource along a path,  which is a sequence
of resources connected by edges, starting from the entry resource.
For example, before entering a room in a hotel, a subject enters the
hotel's lobby from the street, and then goes through the corridor.
Figure~\ref{fig:synth-overview}(c) gives an example of a resource
structure. 


\para{Configurations}
Each edge of a resource structure represents a gate controlled
by a local policy installed on the gate's PDP.  We therefore
define a \emph{configuration} for a resource structure~$S$ as a
function that assigns to each edge of~$S$ a local policy.  
We write $C_S$ for the set of all configurations for~$S$. 
The set~$C_S$ is partially ordered under the relation $\sqsubseteq_S$, defined as:~$c\sqsubseteq_S~\!\!c'$ if for any edge~$e$ of~$S$ we have $c(e)\subseteq c'(e)$. 
Namely, a configuration is less permissive than another configuration if for any edge the former assigns a less permissive local policy than the latter.

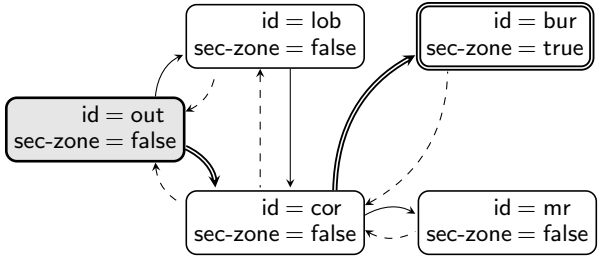
\begin{figure}[t]
\centering
\small
\begin{tikzpicture}[->,>=stealth,shorten >=1pt,auto]
  \node[entry] (out)     at (0,0) {\begin{tabular}{rl}$\sf id$ & $\sf = out$\\$\sf \operatorname{\sf sec-zone}$ & $\sf =false$\end{tabular}};
  \node[resource, anchor=west] (lob)  at ($(out.north east) + (0, 0.8)$) {\begin{tabular}{rl}$\sf id$ & $\sf = lob$\\$\sf \operatorname{\sf sec-zone}$ & $\sf =false$\end{tabular}};
  \node[resource, anchor=west] (cor)  at ($(out.south east) - (0, 0.8)$) {\begin{tabular}{rl}$\sf id$ & $\sf = cor$\\$\sf \operatorname{\sf sec-zone}$ & $\sf =false$\end{tabular}};
  \node[resource, anchor=west] (mr)  at ($(cor.east) + (0.7, 0)$) {\begin{tabular}{rl}$\sf id$ & $\sf = mr$\\$\sf \operatorname{\sf sec-zone}$ & $\sf =false$\end{tabular}};
  \node[resource, double, line width=0.7pt, anchor=west] (bur) at ($(lob.east) + (0.7, 0)$) {\begin{tabular}{rl}$\sf id$ & $\sf = bur$\\$\sf \operatorname{\sf sec-zone}$ & $\sf =true$\end{tabular}};
  \path 
        ($(cor.east) + (0, 0.1)$) edge[->, bend left] ($(mr.west) + (0, 0.1)$)
        ($(cor.east) - (0, 0.1)$) edge[<-, bend right, dashed] ($(mr.west) - (0, 0.1)$)
        ($(cor.north east) - (0.4, 0)$) edge[->, double, line width=0.7pt, bend left] ($(bur.south west) + (0, 0.2)$)
        ($(cor.north east) - (0, 0.2)$) edge[<-, bend right, dashed] ($(bur.south west) + (0.4, 0)$)
        ($(out.south east) + (0, 0.2)$) edge[->, double, line width=0.7pt, bend left] ($(cor.north west) + (0.4, 0)$)
        ($(out.south east) - (0.4, 0)$) edge[<-, bend right, dashed] ($(cor.north west) - (0, 0.2)$)
        ($(out.north east) - (0.4, 0)$) edge[->, bend left] ($(lob.south west) + (0, 0.2)$)
        ($(out.north east) - (0, 0.2)$) edge[<-, bend right, dashed] ($(lob.south west) + (0.4, 0)$)
        ($(lob.south) - (0.2,0)$) edge[<-, dashed] ($(cor.north) - (0.2, 0)$)
        ($(lob.south) + (0.2,0)$) edge[->] ($(cor.north) + (0.2, 0)$);
\end{tikzpicture}
\caption{The double-lined edges denote the PEPs that deny the access request~$q = \{{\sf role} \mapsto  {\sf visitor}, {\sf time}\mapsto 10, \operatorname{\sf correct-pin}\mapsto \bot\}$, given the configuration~$c$ from Figure~\ref{fig:synth-overview}.
The resource structure $S_{c,q}$ is obtained by removing the double-lined edges and nodes.
}
\label{fig:restricted-example}
\end{figure} 

We can now define which resources are accessible given an access
request and a configuration.
For a resource structure~$S$, a configuration~$c$ for~$S$, and an
access request~$q$, we define~$S_{c,q}$ as the resource structure
obtained by removing all the edges from~$S$ whose policies deny~$q$,
and then removing all nodes that are not reachable from the entry
resource.
The structure~$S_{c,q}$'s entry resource is the same as~$S$'s.
To illustrate, consider the resource structure $S$ and the configuration
$c$ given in Figure~\ref{fig:synth-overview}, and the access request
$q = \{{\sf role} \mapsto {\sf visitor}, {\sf time}\mapsto 10,
\operatorname{\sf
  correct-pin}\mapsto \bot\}$. The
side-entrance PEP and the bureau PEP deny $q$ and therefore these two
edges are removed from~$S$. The node that represents the bureau is not reachable from the entry resource and it is thus also removed. In Figure~\ref{fig:restricted-example} we
depict the removed edges and nodes. 

We remark that the structure~$S_{c,q}$ is defined for a fixed access request~$q$. Access requests, which assign values to subject and contextual attributes, can however change, for instance when a subject's role is revoked or as time progresses.
We abstract away such changes in~$S_{c,q}$'s definition. In our running example, this amounts to assuming that a subject's role does not change during this time, and the time needed to move through the office building is negligible compared to the time needed for a subject's access rights to change; for example, the requirements {\bf R1-5} stipulate that subject's access rights may change only twice per day --- at $8$AM and at $8$PM.
This abstraction corresponds to taking a snapshot of all the attributes, and then computing~$S_{c,q}$ based on the snapshot.
We refer to these snapshots as \emph{sessions}.
Henceforth we interpret global requirements and local policies in the context of such sessions.

Interpreting requirements and polices in the context of a session is justified for many practical scenarios. This is because, in most practical settings, changes in subject and contextual attributes are addressed through out-of-band mechanisms.
To illustrate, consider a subject who has the role visitor and enters the meeting room of our running example at $3$PM as permitted by the system's requirements. Now, suppose that the subject's visitor role is revoked at $4$PM, or that the subject remains in the meeting room until $10$PM. No access-control system can force the subject to leave.  In practice, out-of-band mechanisms, such as security guards, address such concerns.

In the following sections, we confine our attention to configurations that do not introduce deadlocks. That is, we consider those configurations~$c$ where for any~$q\in\mathcal{Q}$, the structure~$S_{c,q}$ is deadlock-free. In Section~\ref{sec:gen-req}, we describe how this provision can be encoded as a global requirement.
%


\section{Specifying Requirements}

\label{sec:reqs}

In this section we define \lang, a simple declarative language for specifying requirements.
We give the language's syntax and semantics in Section~\ref{sec:syn-sem}.  To simplify the specification of global requirements,  in Section~\ref{sec:patterns} we present four requirement patterns that capture common access-control idioms for physical spaces. 
Finally, in Section~\ref{sec:gen-req}, we illustrate the specification of two generic access-control requirements: deny-by-default and deadlock-freeness.

\subsection{Requirement Specification Language}
\label{sec:syn-sem}
The design of \lang has been guided by real-world physical access-control requirements. Virtually all such requirements can be formalized as properties that specify which physical spaces subjects can and cannot access, directly and over paths, based on the security context and on the physical spaces they have accessed.
In our physical access-control model, subjects choose which physical spaces to access, which induces a branching structure over the spaces they access.
We therefore build our requirement specification language \lang upon the computation tree logic
(CTL)~\cite{Emerson:1991:TML:114891.114907}, whose branching semantics is a natural fit for physical spaces.

\para{Syntax}
A requirement specified in \lang is a formula of the form
$T\Rightarrow \varphi$ given by the following BNF:
\[ \begin{array}{rcl} 
T & ::= & {\sf true}\mid a_s\in D\mid a_c\in D\mid \neg T \mid T
\wedge T\\ \varphi & ::= & {\sf true}\mid a_r\in D\mid \neg \varphi
\mid \varphi \wedge \varphi\mid {\sf EX}\varphi\mid {\sf AX}\varphi
\\[2pt]
&  \mid & {\sf E}[\varphi {\sf U} \varphi]\mid {\sf A}[\varphi {\sf U}
\varphi]~. 
\end{array}
\]
Here $a_s\in {\cal A}_S$ is a subject attribute, $a_c\in {\cal A}_C$
is a contextual attribute, $a_r\in {\cal A}_R$ is a resource
attribute, and~$D\subseteq {\cal V}$ is a finite subset of values.
The formula $T$ is a constraint over subject and contextual attributes
that defines the access requests to which the requirement applies. We
call $T$ the \emph{target}.  The formula $\varphi$ is a CTL formula
over resource attributes. It defines a path property that must hold
for all access requests to which the requirement is applicable.  We
call $\varphi$ an \emph{access constraint}.  

Note that additional Boolean and 
CTL operators can be defined in the standard
way.  For example, we write~$\mathsf{false}$ for~$\neg \mathsf{true}$, and
define the Boolean connectives~$\vee$ and~$\Rightarrow$ in the
standard manner using~$\neg$ and~$\wedge$.
We will later make use of the CTL operators ${\sf EF}\varphi$, ${\sf AG}\varphi$,
and ${\sf A}[\varphi{\sf R}\psi]$, which are defined as ${\sf
E}[{\sf true}~{\sf U}~\varphi]$, $\neg ({\sf EF} \neg \varphi)$, and $\neg ( {\sf E}[\neg \varphi{\sf
U}\neg\psi])$, respectively.  
Below we give intuitive explanations of~$\sf EX$, $\sf AX$, $\sf EU$, and $\sf AU$,
which are standard CTL connectives.

\begin{figure}[t]
	\centering
	\tabcolsep=4pt
	\fbox{\begin{tabular}{rcl}
			$a = c$ & $:=$ & $a \in \{c\}$\\
			$a \neq c$ & $:=$ & $\neg (a  = c)$\\
			$a_{\sf bool}$ & $:=$ & $a_{\sf bool} = {\sf true}$\\
			$a_{\sf num} \leq n$ & $:=$ & $a_{\sf num} \in \{0, \ldots, n\}$\\
			$a_{\sf num} \geq n$ & $:=$ & $\neg (a_{\sf num} \leq n - 1)$\\
			$n \leq a_{\sf num}\leq  n'$ & $:=$ & $(a_{\sf num}\geq n) \wedge (a_{\sf num}\leq n')$\\
		\end{tabular}}
		\caption{Syntactic shorthands: $a\in {\cal A}$ is an attribute,
			$a_{\sf num}\in {\cal A}_{\sf num}$ is a numeric attribute, 
			$a_{\sf bool}\in {\cal A}_{\sf bool}$ is a boolean attribute, $n, n'\in
			\mathbb{N}$ are natural numbers.}
		\label{fig:syn-ext}
	\end{figure}

The 
connectives {\em exists-next} $\sf EX$ and {\em always-next} $\sf AX$
constrain the physical spaces that a subject can access next. In our
running example, suppose that a subject has entered the lobby. The
subject can next enter the corridor or go to the public space: these
are immediately accessible from the lobby.
In the lobby,~${\sf EX} \varphi$ states that the formula~$\varphi$
is true in at least one of these ``next'' spaces. In
contrast,~${\sf AX} \varphi$ states that~$\varphi$ is true both
in the corridor and in the public space.

The operators {\em exists-until} $\sf EU$ and {\em always-until} $\sf
AU$ relate two access constraints $\varphi_1$ and $\varphi_2$ over
paths. The formula~$\sf E[\varphi_1 U\varphi_2]$ states that there
exists a path that reaches a
resource~$r$ that satisfies~$\varphi_2$, and any resource prior to~$r$
on the path satisfies~$\varphi_1$. We use this connective to
formalize, for example, waypointing requirements such as: visitors cannot
access the meeting room until they have accessed the lobby.
The formula~$\sf A[\varphi_1 U\varphi_2]$ states that
every path reaches some resource~$r$
that satisfies~$\varphi_2$, and that any resource prior to~$r$ on the
path satisfies~$\varphi_1$.

To simplify writing attribute constraints in~\lang, 
we introduce in Figure~\ref{fig:syn-ext} abbreviations 
for numeric and boolean attributes.
Based on the attributes' domains, we partition the set of
attributes~$\cal A$ into {\em numeric attributes}~$A_{\sf num}$,
{\em boolean attributes}~$A_{\sf bool}$, and
{\em enumerated attributes}~$A_{\sf enum}$: An attribute~$a$ is {\em
  numeric} if ${\sf dom}(a)= \mathbb{N}\cup \{\bot\}$; it is
{\em boolean} if ${\sf dom}(a) = \{{\sf false}, {\sf true}, \bot\}$; otherwise, it
is enumerated and ${\sf dom}(a)$ is finite.
%
We may write~$a_{\sf num}$ or $a_{\sf bool}$ to emphasize that an
attribute $a$ is numeric or boolean, respectively.

\begin{figure}[t]
	\centering
	\[\begin{tabular}{llll}
	$S, r_0$ & $\models {\sf true}$\\[2pt]
	$S, r_0$ & $\models a\in D$ & \ if~$\quad$ & $L(r_0)(a)\in D$\\[2pt]
	$S, r_0$ & $\models \neg \varphi$ & \ if~$\quad$ & $S, r_0 \not\models \varphi$\\[2pt]
	$S, r_0$ & $\models \varphi_1 \wedge \varphi_2$ & \ if~$\quad$ & $S, r_0 \models \varphi_1\ \text{and}\ S, r_0\models \varphi_2$\\[2pt]
	$S, r_0$ & $\models {\sf EX}\varphi$ & \ if~$\quad$ & $\exists (r_0, r_1, \cdots)\in S(r_0).\ S, r_1\models \varphi$ \\[2pt]
	$S, r_0$ & $\models {\sf AX}\varphi$ & \ if~$\quad$ & $\forall (r_0, r_1, \cdots)\in S(r_0).\ S, r_1\models \varphi$ \\[2pt]
	$S, r_0$ & $\models {\sf E}[\varphi_1 {\sf U}\varphi_2]$ & \ if~$\quad$ & $\exists (r_0, r_1, \cdots)\in S(r_0).\ \exists i\geq 0.$\\ [2pt]
	&&& $S, r_i\models \varphi_2 \wedge\forall j\in [0, i).\ S, r_j\models \varphi_1$\\
	$S, r_0$ & $\models {\sf A}[\varphi_1 {\sf U}\varphi_2]$ & \ if~$\quad$ & $\forall(r_0, r_1, \cdots)\in S(r_0).\ \exists i\geq 0.$\\[2pt]
	&&& $S, r_i\models \varphi_2 \wedge\forall j\in [0, i).\ S, r_j\models \varphi_1$
	\end{tabular}\]
	\caption{The  relation~$\models$ between a resource structure $S = ({\cal R}, E, \entry, L)$, a resource $r_0\in {\cal R}$, and an access constraints~$\varphi$.}
	\label{fig:semantics}
\end{figure}

\para{Semantics}
%
We first inductively define the
satisfaction relation~$\vdash$ between an access request~$q\in \mathcal{Q}$
and a target:
\[
\begin{array}{rclcl}
q & \vdash & {\sf true}\\
q & \vdash & a\in D & \text{if} & q(a)\in D\\
q & \vdash & \neg T & \text{if} & q\not\vdash T\\
q & \vdash & T_1\wedge T_2 & \text{if} & q\vdash T_1\ \text{and}\ q\vdash T_2~.
\end{array}
\]
A requirement $T\Rightarrow \varphi$ is {\em applicable} to an access
request~$q$ iff $q$ satisfies the target $T$, i.e.\ $q\vdash T$. 
For example, the requirement $({\sf role = visitor})\Rightarrow
\varphi$ is applicable to all access requests that assign the value
$\sf visitor$ to the subject attribute $\sf role$.

\begin{figure*}
\centering
\small
\setlength{\tabcolsep}{7pt}
\begin{tabular}{lllp{0.28\textwidth}l}
\toprule
{\bf Pattern} & {\bf Shorthand} & {\bf Specification} & {\bf Description} & {\bf Intuitive Semantics}
\\
\midrule
Permission & 
$T\Rightarrow \textsc{Grant}(\varphi)$ & $T\Rightarrow {\sf EF}\ \varphi$ & $T$-requests can access $\varphi$-spaces. & 
\adjustbox{valign=b}{\begin{tikzpicture}[->,shorten >=1pt,auto]
\small
  \node[entry, minimum height=12pt, minimum width=12pt] (perm) at (0,0) {$\entry$};
  \node[resource, minimum height=12pt, minimum width=12pt] (a) at ($(perm) + (1, 0)$) {$\varphi$};
  \draw[->] (perm)-- node {\cmark} (a);
\end{tikzpicture}}
\\
Prohibition & 
$T\Rightarrow \textsc{Deny}(\varphi)$ & $T\Rightarrow {\sf AG }(\neg
\varphi)$ & $T$-requests cannot access $\varphi$-spaces.  & 
\adjustbox{valign=b}{\begin{tikzpicture}[->,shorten >=1pt,auto]
\small
  \node[entry, minimum height=12pt, minimum width=12pt] (proh) at ($(perm) - (0, 2)$) {$\entry$};
  \node[resource, minimum height=12pt, minimum width=12pt] (a) at ($(proh) + (1, 0)$) {$\varphi$};
  \draw[->] (proh)-- node {\xmark} (a);
\end{tikzpicture}}
\\
Blocking & 
$T\Rightarrow \textsc{Block}(\varphi, \psi)$ & $T\Rightarrow {\sf AG}(
\varphi \Rightarrow {\sf AG}(\neg \psi))$ &  $T$-requests cannot access a
$\psi$-space after accessing a $\varphi$-space. & 
\adjustbox{valign=c}{\begin{tikzpicture}[->,shorten >=1pt,auto]
\small
  \node[entry, minimum height=12pt, minimum width=12pt] (block) at ($(proh) + (2.5, 0)$) {$\entry$};
  \node[resource, minimum height=12pt, minimum width=12pt] (a) at ($(block) + (1, 0)$) {$\varphi$};
  \node[resource, minimum height=12pt, minimum width=12pt] (b) at ($(block) + (2, 0)$) {$\psi$};
  \draw[->] (block)--  (a);
  \draw[->] (a)-- node {\xmark} (b);
\end{tikzpicture}}
\vspace{-4pt}\\
Waypointing & 
$T\Rightarrow \textsc{Waypoint}(\varphi, \psi)$ & $T\Rightarrow
{\sf A}[\varphi{\sf R}\psi]$ & $T$-requests must access
a $\varphi$-space before accessing a $\psi$-space. & 
\adjustbox{valign=c}{\begin{tikzpicture}[->,shorten >=1pt,auto]
\small
  \node[entry, minimum height=12pt, minimum width=12pt] (way) at ($(block) + (3.5, 0)$) {$\entry$};
  \node[resource, minimum height=12pt, minimum width=12pt] (a) at ($(way) + (1, 0)$) {$\varphi$};
  \node[resource, minimum height=12pt, minimum width=12pt] (b) at ($(way) + (2, 0)$) {$\psi$};
  \draw[->] (way)-- (a);
  \draw[->] (a)-- (b);
  \draw[->] (way) to [bend left=35] node {\xmark} (b);
\end{tikzpicture}}
\\
\bottomrule
\end{tabular}
\caption{\lang Patterns: The entry resource~$\entry$ in the intuitive semantics is depicted using a gray rectangle. The arrows $\varphi \xrightarrow{\text{\cmark}} \psi$ $\big(\varphi \xrightarrow{\text{\xmark}} \psi\big)$ indicate that there must (must not) exist a path from a $\varphi$-space to a $\psi$-space along which $T$-requests are granted.}
\label{fig:patterns}
\end{figure*}

Let $S = ({\cal R}, E, \entry, L)$ be a resource structure.  A {\em path}
of $S$ is an infinite sequence of resources~$(r_0, r_1, \cdots)$ such
that $\forall i\geq 0.\ (r_i, r_{i+1})\in E$, and we denote the set of
all paths rooted at a resource $r_0$ by $S(r_0)$. 
In Figure~\ref{fig:semantics}, 
we inductively define the satisfaction relation~$\models$ between a
resource structure, a resource, and an access constraint.
A resource structure $S$ with an entry resource $\entry$ {\em satisfies} an
access constraint~$\varphi$, denoted by $S \models \varphi$, iff $S,
\entry\models \varphi$.
\begin{definition}
Let $S$ be a resource structure, $c$ a configuration for $S$, and
$T\Rightarrow \varphi$ a requirement. $S$ configured with $c$ {\em
  satisfies} $T\Rightarrow \varphi$, denoted by $S, c\Vdash
(T\Rightarrow \varphi)$, iff $q\vdash T$ implies
$S_{c,q}\models\varphi$, for any access request $q\in\mathcal{Q}$.
\end{definition}

We extend $\Vdash$ to sets of requirements as expected. Given a set of
requirements $R = \{T_1\Rightarrow \varphi_1, \ldots, T_n\Rightarrow
\varphi_n\}$, a resource structure $S$ configured with $c$ satisfies
$R$, denoted by $S, c\Vdash R$, iff $S, c\Vdash (T_i \Rightarrow \varphi_i)$ for
all~$i$, $1\leq i \leq n$.

We remark that resource structures can easily be represented using
standard Kripke structures~\cite{Emerson:1991:TML:114891.114907} by
mapping each resource to a Kripke state and each resource attribute
valuation to sets of atomic propositions. The access constraints can
be similarly mapped to standard CTL formulas by translating attribute
constraints into propositional logic.
Note however that while Kripke structures are often used to represent
changes of, say, a concurrent system's state over time, resource
structures model static physical spaces.

\subsection{Requirement Patterns}
\label{sec:patterns}

\lang can be directly used to specify global requirements.  However, to illustrate its use and expressiveness, we present the formalization of common physical access-control idioms.


We have studied the requirements of an airport, a corporate building,
and a university campus to elicit the common structure of physical
access-control requirements.  To distill the basic requirement
patterns, we split complex requirements into their atomic parts. Our
analysis revealed four common patterns, which we formalize below.  The
first pattern abstracts {\em positive} requirements, which
stipulate that the access-control system must grant certain access
requests.  The remaining three patterns capture {\em negative}
requirements, which stipulate that the access-control system must deny
certain access requests.

We use the following terminology when describing requirements.
Given a target $T$, we call an access request $q$ a {\em
$T$-request} if $q\vdash T$, i.e.\ $q$ satisfies the target $T$. 
Given a resource structure $S$ and an access constraint $\varphi$, we say
that a subject can access a {\em $\varphi$-space} of $S$ if the
subject can access a physical space $r_0$ of $S$ such that $S, r_0\models
\varphi$, i.e. the space $r_0$ satisfies the access constraint~$\varphi$. 
Our patterns are summarized in Figure~\ref{fig:patterns}.

\para{Permission}
The {\em permission} pattern abstracts requirements stating that
$T$-requests can access $\varphi$-spaces from the entry resource.
Permission requirements have the form $T\Rightarrow ({\sf EF}\
\varphi)$.  The exists-future operator $\sf EF$ formalizes
that a $\varphi$-space is reachable from 
the entry resource.
For example, the requirement {\bf R3} stipulating that employees can access the
bureau between $8$AM and $8$PM is formalized as 
\[ \big( ({\sf role} = {\sf employee}) \wedge (8\leq {\sf time}\leq 20)\big) \Rightarrow {\sf EF}\ ({\sf id} = {\sf bur}).  \]
The target $({\sf role} = {\sf
employee}) \wedge (8\leq {\sf time}\leq 20)$ formalizes that this
requirement is applicable only to access requests made by visitors at
times between $8$AM and $8$PM. The access constraint ${\sf EF}({\sf id} =
{\sf bur})$ is satisfied iff the resource structure has a path from the
entry resource to the bureau.   
The requirements {\bf R1} and {\bf R4} of our running example are also
instances of the permission pattern.

%
%
%
\para{Prohibition}
Dual to the permission pattern, the {\em prohibition} pattern captures
requirements stating that $T$-access requests cannot access a
$\varphi$-space.
Prohibition requirements have the form $T\Rightarrow {\sf AG}(\neg
\varphi)$. The operator $\sf AG$ quantifies over
all paths reachable from the entry resource.
An example taken from our airport requirements is: Passengers cannot
access the departure gate zones
without a boarding pass. Another example is 
requirement~{\bf R5}, formalized as
\[ ({\sf role}\neq {\sf employee})\Rightarrow {\sf AG}(\neg 
\operatorname{\sf sec-zone})~. \]
The target ${\sf role}\neq {\sf employee}$ is satisfied by access
requests that assign a value other than $\sf employee$ to the
attribute $\sf role$. The access constraint ${\sf AG}(\neg 
\operatorname{\sf sec-zone})$ is satisfied if no path leads to a
security zone.

\para{Blocking}
The {\em blocking} pattern captures requirements stating that subjects
cannot access a $\psi$-space after they have accessed a
$\varphi$-space. 
Intuitively, accessing a $\varphi$-space \emph{blocks} the subject from
accessing $\psi$-spaces. 
At international airports, for example, passengers may not access
departure gate zones after they have accessed the baggage claim. 
Blocking requirements have the form $T\Rightarrow {\sf AG} (\varphi
\Rightarrow {\sf AG}(\neg \psi))$.
The airport example is formalized as: 
\begin{align*} ({\sf role} = {\sf passenger}) \Rightarrow {\sf AG} 
& \big( ({\sf zone} = \operatorname{\sf baggage-claim}) \\ & \Rightarrow {\sf AG}\
\neg (\sf zone = {\sf departure})\big)~.  \end{align*}
This requirement instantiates the blocking pattern: the target $T$ is
$({\sf role} = {\sf passenger})$, and the two access constraints~$\psi$ and~$\varphi$ are $(\operatorname{\sf zone} = {\sf departure})$
and $({\sf zone} = \operatorname{\sf baggage-claim})$.

\para{Waypointing}
The {\em waypointing} pattern captures requirements stipulating that
subjects must first access a $\varphi$-space before accessing a
$\psi$-space. For example, passengers cannot access an airport's terminal before
they have passed through a security check. This is a negative
requirement that restricts how passengers can access the terminal. 
Waypointing requirements have the form $T\Rightarrow ({\sf A}[\varphi{\sf R}\psi])$.
The globally-release operator $\sf AR$ quantifies over all
paths from the entry resource and formalizes that  if
$\psi$ holds at some point,
then $\varphi$ was valid at least once beforehand.
The requirement \textbf{R2} of our running example is an instance of
the waypointing pattern and is formalized as
\[ ({\sf role} = {\sf visitor}) \Rightarrow {\sf A}[({\sf id} = {\sf lob}){\sf
R}({\sf id} = {\sf mr})]~.  \]
The target specifies that this requirement applies to all access
requests made by visitors. The access constraint is satisfied if all
paths to the meeting room go through the lobby.

The four idioms just described cover \emph{all} the requirements that arose in the case studies 
that we report on in Section~\ref{case-studies-sec}. 
We remark though that there are global requirements that are not instances of these four patterns. 
For example, in corporate buildings,  a subject must be able to access the parking lot if he or she has access to an office.
Although this requirement cannot be expressed using the above patterns, it can be directly formalized in \lang as follows:
\begin{align*} 
{\sf true} \Rightarrow \big( ({\sf EF} ({\sf zone} = {\sf office})) \Rightarrow ({\sf EF} ({\sf id}
=\operatorname{\sf parking-lot}))\big)~.
\end{align*}
In general, as \lang supports all CTL operators, it can specify any branching property expressible in CTL.

\subsection{Generic Requirements}
\label{sec:gen-req}

We now describe two commonly-used generic requirements.

\para{Deny-by-default}
The deny-by-default principle stipulates that if an access request can be denied without violating the requirements, then it should be denied; cf.~\cite{saltzer75}.
Security engineers often follow this principle to avoid overly permissive local policies.
To illustrate, consider our running example and imagine that the role $\sf intern$ is contained in the domain of the attribute $\sf role$.
The requirements given in Figure~\ref{fig:synth-overview}(b) do not prohibit an intern from accessing, say, the meeting room. However, denying interns access to the meeting room is also compliant with these requirements.

%

The following requirement, called {\em deny-by-default}, instantiates the above principle: If no positive requirement is applicable to an access request, then only the entry space is accessible to the subject who makes such a request. 
%
To formalize this requirement, we first define positive and negative requirements. 
Let $c$ and $c'$ be two configurations for a given resource structure~$S$.
A requirement $T\Rightarrow \varphi$ is \emph{positive} if $S,c\Vdash (T\Rightarrow \varphi)$ and $c \sqsubseteq_S c'$ imply  $S,c'\Vdash (T\Rightarrow \varphi)$.
A requirement $T\Rightarrow \varphi$ is \emph{negative} if $S,c\Vdash (T\Rightarrow \varphi)$ and $c' \sqsubseteq_S c$ imply  $S,c'\Vdash (T\Rightarrow \varphi)$.
Intuitively, if a configuration satisfies a positive (negative) requirement, then any more (less) permissive configuration also satisfies the requirement. 
We remark that although not all requirements are positive or negative, most real-world requirements are, including all requirements specified in this paper.

Let $R$ be a set of requirements that contains only positive and negative requirements, and let 
$\{T_1\Rightarrow \varphi), \cdots, T_n\Rightarrow \varphi_n\}$ be the set of all positive requirements contained in $R$. 
The deny-by-default requirement for $R$ is
\[ (\neg T_1)\wedge \cdots \wedge (\neg T_n)\Rightarrow {\sf AX}\ ({\sf id} = {\sf entry})~. \]
Here we assume that $L(\entry)({\sf id}) = {\sf entry}$, i.e. the entry resource $\entry$ is labeled with $\sf entry$. 
Adding this requirement to our running example's requirements would ensure that an intern cannot access, e.g., the meeting room.

\para{Deadlock-freeness}
A {\em deadlock-freeness} requirement stipulates that there are no deadlocks in a system, i.e.\ resources that a subject can access and then never leave. For example, the meeting room of our running example would be a deadlock if visitors could enter it, but never leave.
As discussed in our system model, local policies that introduce deadlocks are undesirable.

%
Formally, 
the deadlock-freeness requirement is defined as:
\[ {\sf true}\Rightarrow {\sf AG}\ {\sf EX}\ {\sf true}~. \]
This requirement applies to all access requests. The access constraint ${\sf AG}\ {\sf EX}\ {\sf true}$ states that for any resource a subject can access, there is a resource that the subject can access next.
A resource structure~$S$ and a configuration $c$ satisfy this requirement iff for any access requests $q\in {\cal Q}$, $S_{c,q}$ has no deadlocks.
\section{Policy Synthesis Problem}
\label{sec:synthesis}

We now define the policy synthesis problem.  
We show that this problem is decidable but NP-hard. 
%
%


\subsection{Problem}

\label{sec:problem}

\begin{definition}
The {\em policy synthesis problem} is as follows:\\[4pt]
\begin{tabular}{p{0.15\columnwidth}p{0.8\columnwidth}}
{\bf Input.} & 
A resource structure $S$ and a set of requirements~$R$.\\
{\bf Output.} & A configuration $c$ such that $S, c\Vdash R$,
if such a configuration exists, and $\sf unsat$ otherwise.
\end{tabular}
\end{definition}

The synthesized configuration defines the local policies to be deployed at
the PEPs. 
Recall that a policy is extensionally defined as the set of access
requests for which the PEP grants access. As such a set may, in
general, be infinite, one cannot simply output an extensional definition of the
synthesized configuration. 
We therefore define local policies intensionally by constraints over
subject and contextual attributes, expressed in the same language that we specify requirement targets in~\secref{sec:reqs}. 
The semantics of an intensional local policy $P$ is then simply 
$\lambda q.\ \text{if}\ q\vdash P\ \text{then}\ {\sf grant}\ \text{else}\ {\sf deny}$.
%
Figure~\ref{fig:synth-overview} illustrates the input and output to the policy synthesis problem for our running example.

An example of a local policy defined over the attributes $\sf role$ and $\sf time$ is 
$({\sf role} = {\sf visitor}) \wedge (8 \leq {\sf time}\leq 20)$.
This local policy grants all access requests that assign the value $\sf
visitor$ to the attribute $\sf role$ and a number between $8$ and $20$ to
the attribute $\sf time$. Note that this local policy is also the target of requirement~\req{1}.

\subsection{Decidability}

\label{sec:decidability}


To show that the policy synthesis problem is decidable, we give a synthesis algorithm, called $\CS$, that uses  controller synthesis as a subroutine. 
In the following, we first define the controller synthesis problem.
We then show how the algorithm $\CS$ constructs the PEPs' local policies by solving multiple controller synthesis instances. 

\para{Controller Synthesis Problem}
Controller synthesis algorithms take as input a description of an uncontrolled system, called a plant, along with a specification, and output a controller that restricts the plant so that it satisfies the given specification. In our setting, the plant is the resource structure and the specification is an access constraint, i.e.\ a CTL formula over resource attributes. The synthesized controller then defines which PEPs must grant or deny the access request so that the access constraint is satisfied.
For simplicity, we do not define the controller synthesis problem in its most general form. For our needs the following simpler definition suffices.

\begin{definition}
The {\em controller synthesis problem} is as follows:\\[4pt]
\begin{tabular}{p{0.15\columnwidth}p{0.8\columnwidth}}
{\bf Input.} & A resource structure $S = ({\cal R}, E, \entry, L)$ and an access constraint $\varphi$.\\
{\bf Output.} &  A set $E'\subseteq E$ of edges such that $({\cal R}, E', \entry, L)\models
\varphi$, if such an $E'$ exists, and $\sf unsat$ otherwise.
\end{tabular}
\end{definition}

The controller synthesis problem can be reduced to synthesizing a memoryless controller for a Kripke structure given a CTL specification. Deciding whether a controller synthesis instance has a solution is NP-complete~\cite{Antoniotti:1995}. Systems such as MBP~\cite{Bertoli01b} can be used to synthesize controllers. For a comprehensive overview of controller synthesis see~\cite{Ramadge:1987:SCC:35469.35482}.

\para{Algorithm}
The algorithm $\CS$ is based on two insights. First, for a given access request $q$, we can use controller synthesis to identify which PEPs must grant or deny~$q$. 
In more detail, we can compute $({\cal R}, E', \entry, L)\models \varphi_q$, where $\varphi_q$ conjoins all access constraints of the requirements that are applicable to~$q$.
The edges in $E'$ represent the PEPs that must grant~$q$ and those in $E\setminus E'$ the PEPs that must deny~$q$.
A configuration can thus be synthesized by solving one controller synthesis instance for each access request. However, there are infinitely many access requests.
Our second insight is that we can construct a configuration by solving finitely many controller synthesis instances. 
We partition the set~$\cal Q$ of access requests into~$2^{|R|}$ equivalence classes, where two access requests are equivalent if the same set of requirements are applicable to them.
Solving one controller synthesis instance for one representative access request per equivalence class is sufficient for our purpose.

\begin{algorithm}[t]
	\DontPrintSemicolon
	\KwIn{Resource stricture $S = ({\cal R}, E, \entry, L)$, \hspace{50pt}a set of requirements~$R$}
	\KwOut{A configuration $c$ or ${\sf unsat}$}
	\Begin{
		\For {$e\in E$} {\label{line:foredges}
			$c(e) \gets {\sf true}$\;\label{line:setedgetotrue}
		}
		\For{$R' \subseteq R$} {\label{line:for-loop}
			$T \gets T_1\wedge \cdots \wedge T_i \wedge \neg T_{i+1}\wedge
			\cdots \wedge \neg T_n$, where\;\label{line:compute-target}
			\hspace{10pt} $\{ T_1\Rightarrow \varphi_1, \ldots, T_i\Rightarrow \varphi_i\} = R'$ and\;
			\hspace{10pt} $\{ T_{i+1}\Rightarrow \varphi_{i+1}, \ldots, T_n\Rightarrow \varphi_n\} = R\setminus R'$\;
			\label{alg:cs:target}
			\If {$\exists q\in {\cal Q}.\ q\vdash T$} { \label{line:cs-sat}
				$\varphi \gets \varphi_1\wedge \cdots\wedge \varphi_i$\;\label{alg:cs:varphi}
				\If {${\sf cs}(S, \varphi) = {\sf unsat}$} { \label{line:cs-synth}
					\Return $\sf unsat$\;
				} \Else {
				\For {$e\in E\setminus {\sf cs}(S, \varphi)$} { \label{alg:cs:edges}
					$c(e) \gets c(e) \wedge (\neg T)$\;\label{alg:cs:conf}
				}
			}
		}
	}
	\Return $c$\;
}
\caption{The algorithm ${\cal S}_{\sf cs}$ for synthesizing policies
	using controller synthesis. The controller synthesis algorithm,
	denoted~${\sf cs}(S, \varphi)$, outputs either a subset of~$E$ or $\mathsf{unsat}$.}
\label{alg:cs}
\end{algorithm}

The main steps of the algorithm $\CS$ are given in Algorithm~\ref{alg:cs}. 
The algorithm iteratively constructs a configuration $c$ as follows. 
Initially, it sets all local policies to $\sf true$ (lines~\ref{line:foredges}-\ref{line:setedgetotrue}). The algorithm iterates over all subsets $R' = \{T_1\Rightarrow\varphi_1, \ldots,
T_i\Rightarrow \varphi_i\}$ of the requirements~$R$ (line~\ref{line:for-loop}). 
The conjunction $T = T_1\wedge \cdots\wedge T_i \wedge \neg T_{i+1} \wedge \cdots \neg T_n$ constructed at line~\ref{line:compute-target} is satisfied by all access requests to which only the requirements contained in $R'$ are applicable. 
The set~$\{q\in {\cal Q}\mid q\vdash T\}$ is an equivalence class of access requests.
%
If this equivalence class is nonempty, i.e. $\exists q\in {\cal Q}.\ q\vdash T$, then $c$ must grant and deny all access requests contained in it in conformance with the access constraints defined by~$R'$.
Lines~\ref{alg:cs:varphi}-\ref{alg:cs:conf} define how the algorithm~$\CS$ updates~$c$. First, it constructs  the conjunction $\varphi$ of the access constraints defined by the requirement in~$R'$.
It then executes the controller synthesis algorithm, denoted by ${\sf cs}$, with the inputs $S$ and $\varphi$.
If the algorithm $\sf cs$ returns $\sf unsat$,  then the requirements are not satisfiable for the given resource structure, and the algorithm~$\CS$ thus returns $\sf unsat$. Otherwise, the algorithm $\sf cs$ returns a set $E'\subseteq E$ of edges. The algorithm updates the configuration $c$ as follows: for any edge in $E\setminus E'$, the configuration is modified to deny access to all requests in the equivalence class defined by~$R'$.
The algorithm terminates when all subsets of the global requirements have been considered.

\begin{theorem}
Let $S$ be a resource structure and $R$ a set of requirements.
If ${\cal S}_{\sf cs}(S, R) = c$ then $S,c\Vdash R$. 
If ${\cal S}_{\sf cs}(S, R) = {\sf unsat}$ then there is no configuration $c$ such that $S, c\Vdash R$.
\label{thm:decompose} 
\end{theorem}

\ifext
We prove this theorem and give the complexity of~$\CS$ in Appendix~\ref{sec:decidability}.
\else
We prove this theorem and give the complexity of~$\CS$ in~\cite{tech-report}.
\fi

\para{Example} To illustrate~${\cal S}_{\sf cs}$,
consider our running example and the requirements {\bf R2} and {\bf
R5} formalized as follows:
\begin{align*}
\textbf{R2} &:= ({\sf role} = {\sf visitor})\Rightarrow ({\sf A}[
({\sf id} = {\sf lob})\ {\sf R}\ ({\sf id} = {\sf mr})]) \\
\textbf{R5} &:= ({\sf role} \neq {\sf employee})\Rightarrow ({\sf AG}\
\neg \operatorname{\sf sec-zone})~.
\end{align*}
We remark that ${\sf dom}({\sf role}) = \{\bot, {\sf visitor}, {\sf employee} \}$, and therefore the targets ${\sf role} = {\sf visitor}$ and ${\sf role} \neq {\sf employee}$ are not equivalent.
To synthesize a configuration, the algorithm ${\cal S}_{\sf cs}$ executes
the second for-loop four times. Let the
selected subset of requirements in the first iteration be $\{\textbf{R2}, \textbf{R5}\}$.
The conjunction~$T$ of the targets is $({\sf role} = {\sf
  visitor})\wedge ({\sf role} \neq {\sf employee})$, which is
 equivalent to $({\sf role} = {\sf visitor})$.  Hence,~$T$ is
satisfiable.  The access constraint $\varphi$ (see
Algorithm~\ref{alg:cs}, line \ref{alg:cs:varphi}) is then $ ({\sf A}[ ({\sf id} = {\sf
    lob})\ {\sf R}\ ({\sf id} = {\sf mr})]) \wedge ({\sf AG}\ \neg \operatorname{\sf sec-zone})$. A possible output by the controller synthesis
algorithm ${\sf cs}(S, \varphi)$ is $E\setminus \{({\sf cor}, {\sf
  bur}), ({\sf out}, {\sf cor}) \}$. The updated configuration $c$
after the first iteration is therefore
\begin{align*}
c(e)
= \left\{
\begin{array}{ll}
{\sf true}\wedge {\sf role} \neq {\sf visitor} & \text{if}\ e = ({\sf cor}, {\sf bur})\\
{\sf true}\ \wedge {\sf role}\neq {\sf visitor} & \text{if}\ e = ({\sf out}, {\sf cor})\\
{\sf true}& \text{otherwise}~.
\end{array} \right.
\end{align*}
Suppose the outputs to the remaining three controller synthesis instances are
${\sf cs}(S, \varphi_{\{\textbf{R2}\}}) = E\setminus \{ ({\sf out},{\sf cor})\}$, 
${\sf cs}(S, \varphi_{\{ \textbf{R5}\}}) = E\setminus \{ ({\sf cor}, {\sf bur})\}$, and 
${\sf cs}(S, \varphi_\emptyset) =  E$, where $\varphi_X$ denotes
the conjunction of the access constraints of the requirements in $X$.
The simplified configuration $c$ returned by 
${\cal S}_{\sf cs}$ is
\begin{align*}
c(e)
= \left\{
\begin{array}{ll}
{\sf role} = {\sf employee} & \text{if}\ e = ({\sf cor}, {\sf bur})\\
{\sf role}\neq {\sf visitor} & \text{if}\ e = ({\sf out}, {\sf cor})\\
{\sf true}& \text{otherwise}\, .
\end{array} \right.
\end{align*}

\para{Limitations}
The main limitation of the algorithm $\CS$ is that the running time is exponential in the number of requirements, rendering it impractical for nontrivial instances of policy synthesis.  
For example, while the algorithm~$\mathcal{S}_\mathsf{cs}$ takes $2$ seconds to synthesize a configuration for our running example, it does not terminate within an hour for our case studies, reported in Section~\ref{case-studies-sec}.
We give a practical policy synthesis algorithm based on SMT solving in Section~\ref{sec:algorithm}. 

\subsection{NP-hardness}
To show NP-hardness, we reduce propositional satisfiability to the policy synthesis problem. It is easy to see that a propositional formula~$\varphi$ can be encoded, in logarithmic space, as a target~$T_\varphi$ over Boolean attributes.
Consider the policy synthesis problem for the inputs $S$ and $\{(T_\varphi\Rightarrow {\sf false})\}$, where $S$ is an arbitrary resource structure.
If the output to this policy synthesis instance is $\sf unsat$ then for some access request $q$, we have $q\vdash T_\varphi$. Hence $\varphi$ is satisfiable.  Alternatively, the output to the policy synthesis problem is a configuration $c$. Since for any access request $q$ where $q\vdash T_\varphi$ we have $S_{c,q}\models {\sf false}$,  it is immediate that there is no access request $q$ such that $q\vdash T_\varphi$. Therefore, $\varphi$ is unsatisfiable.  



\section{Policy Synthesis Algorithm}

\label{sec:algorithm}

In this section, we define our policy synthesis algorithm based on SMT solving, called $\SMT$. The
algorithm takes as input a resource structure~$S$, a set~$R$ of
requirements, and a set $C$ of configurations. The set $C$ is encoded
symbolically, as we describe shortly. The algorithm outputs a
configuration~$c$ such that $S, c\Vdash R$, if there is such a
configuration in $C$; otherwise, it returns $\sf unsat$. To synthesize
a configuration $c$, the algorithm encodes the question $\exists c\in
C.\ S, c\Vdash R$ in a decidable logic supported by standard SMT
solvers.
Due to its technical nature, we relegate a detailed description of the
encoding to the end of this section.

Our algorithm takes as input a set of configurations, 
and we refer to the symbolic encoding of this set as a {\em
  configuration template}.
The configuration template enables us to restrict the search space:
the algorithm confines its search to the configurations described by
the template.
Our algorithm $\SMT$ is sound, independent of the provided
configuration template. Its completeness, however, depends on
the template.
We show that one can construct a template for which $\SMT$ is
complete, but the resulting template would, in practice, encode so many configurations that the resulting SMT problem would be infeasible to solve.
We therefore strike a balance between the algorithm's completeness and
its efficiency: since real-world local policies often have small
syntactic representations, as demonstrated by our experiments in
Section~\ref{sec:eval}, our policy synthesis tool starts with a
template that defines configurations with succinct local policies, and iteratively
executes $\SMT$, increasing the template's size in each iteration.
It turns out that in our case studies
a small number of iterations is
sufficient to synthesize all local policies. 
Below, we describe the algorithm $\SMT$'s components.

\subsection{Configuration Templates}
A {\em configuration template} assigns to each edge of the resource
structure a symbolic encoding of a set of local policies.
To illustrate this encoding, consider the set of local policies
$\{{\sf true}, {\sf role} = {\sf employee}, {\sf role} \neq {\sf
  visitor}\}$.  We symbolically encode this set for an edge, say $({\sf cor}, {\sf bur})$,
as a constraint over subject and contextual attributes, as well as
a {\em control} variable~$z_{({\sf cor}, {\sf bur})}$:
\begin{equation}\tag{T1}
\label{ex:template}
 \begin{array}{llcl}
\hspace{-10pt}C(({\sf cor}, {\sf bur})) =\hspace{-6pt} & (z_{({\sf cor}, {\sf bur})} = 1 & \hspace{-6pt}\Rightarrow\hspace{-6pt} & {\sf true})\ \wedge\\
& (z_{({\sf cor}, {\sf bur})} = 2 & \hspace{-6pt}\Rightarrow\hspace{-6pt} & {\sf role} = {\sf employee})\ \wedge\\
& (z_{({\sf cor}, {\sf bur})} = 3 & \hspace{-6pt}\Rightarrow\hspace{-6pt} & {\sf role} \neq {\sf visitor})~.
\end{array} 
\end{equation}
The control variable $z_{({\sf cor}, {\sf bur})}$ encodes the choice of one of three local policies for the edge~$({\sf cor}, {\sf bur})$. 
Hence, for this example, the set of configurations defined by the
configuration template contains~$3^{|E|}$ elements, where~$E$ is the
set of edges in the resource structure.
Note that for a set of local policies of size~$n$ (here $n=3$),~$\lceil \log n\rceil$ propositional variables are sufficient for
representing each edge's control variables. To avoid clutter, we will write $C_{r_0, r_1}$ for $C((r_0, r_1))$.

We remark that configuration templates can be used to restrict the search space of configurations to those that satisfy {\em attribute availability} constraints, which restrict the set of attributes that PEPs can retrieve. 
Suppose that only the side-entrance door of our running example is equipped with a keypad.
To account for this constraint, we will restrict the configurations in the template to those that use the $\operatorname{\sf correct-pin}$ attribute only in the local policy of side entrance's lock.

\subsection{Algorithm}
The main steps of~$\SMT$ are given in Algorithm~\ref{alg:smt}. We
describe the algorithm with an example: the input to the algorithm
consists of the resource structure and the requirements $\mathbf{R2}$
and $\mathbf{R5}$ of our running example, along with the above
configuration template~$C$, which maps edges to the set of local
policies $\{{\sf true}, {\sf role} = {\sf employee}, {\sf role} \neq
{\sf visitor}\}$.
The algorithm starts by creating for each requirement a constraint
that asserts the satisfaction of the requirement in the resource
structure, given the template. This constraint is called~$\psi$ in the
algorithm, and is expressed in the logic of an SMT solver. 
This step is implemented by the subroutine~\textsc{Encode},
defined in Figure~\ref{fig:encode}. To encode the satisfaction of access constraints, we follow the standard model-checking algorithm for CTL based on labeling~\cite{Huth:2004:LCS:975331}; we explain this encoding at the end of this section.

As an example, the result of $\textsc{Encode}(S, \textbf{R2}, C)$,
after straightforward simplifications, is the following constraint:
\begin{align*} 
\psi_{R2} := {\sf role} = {\sf visitor} \Rightarrow (\neg C_{\sf out, cor}\vee \neg C_{\sf cor, mr}) \, .
\end{align*}
Here ${\sf role}$ is an \emph{attribute variable}, originating from
\textbf{R2}'s target, and $C_{\sf out, cor}$ and $C_{\sf cor, mr}$
are the symbolic encodings of the local policies for the edges~$({\sf out, cor})$ and $({\sf cor, mr})$, respectively. 
This constraint states that if the
requirement's target ${\sf role} = {\sf visitor}$ is satisfied, then
one of the PEPs along the path that starts at the entry resource and
reaches the meeting room directly through the corridor must deny
access.
Similarly, $\textsc{Encode}(S, \textbf{R5}, C)$ returns the
constraint:
\begin{align*}
\psi_{R5}  :=\ & {\sf role} \neq {\sf employee} \Rightarrow  \\
& ((\neg C_{\sf out, cor}\vee \neg C_{\sf cor, bur}) \\
& \wedge (\neg C_{\sf out, lob}\vee \neg C_{\sf lob, cor} \vee \neg C_{\sf cor, bur})) \, .
\end{align*}
This states that any access request that maps the attribute
$\sf role$ to a value other than $\sf employee$ must be denied by 
at least one PEP along the path to the bureau that goes directly through the corridor, and moreover it must be denied by 
at least one PEP along the path that passes through the lobby.

\begin{algorithm}[t]
	\DontPrintSemicolon
	\KwIn{A resource structure $S = ({\cal R}, E, r, L)$,\hspace{45pt}
		a set~$\{R_1, \cdots, R_n\}$ of requirements,\hspace{55pt}
		a configuration template~$C$}
	\KwOut{A configuration~$c$ or ${\sf unsat}$}
	\Begin{
		$\phi \gets {\sf true}$\;
		\For {$R \in \{R_1, \cdots, R_n\}$} {
			$\psi\gets \textsc{Encode}(S, R, C)$\; \label{line:smt-encode}
			$\phi \gets \phi \wedge \psi$\;
		}
		\If {$(\exists{\vec{z}}. \forall{\vec{a}}.\ \phi)\ \mathrm{is}\ 
			{\sf sat}$} { \label{line:smt-solve}
			${\cal M}\gets \textsc{Model}(\exists{\vec{z}}. \forall{\vec{a}}.\ \phi)$\;
			\For {$e \in E$} {
				$c(e) \gets \textsc{Derive}(C(e), {\cal M})$\;
			}
			\Return{$c$}\;
		} \Else {
		\Return{${\sf unsat}$}
	}
}
\caption{The algorithm ${\cal S}_{\sf smt}$ for synthesizing policies
	using SMT solving.}
\label{alg:smt}
\end{algorithm}

\begin{figure}
\fbox{\begin{minipage}{0.98\columnwidth}
$\textsc{Encode}(S, T\Rightarrow \varphi, C)\ \textbf{returns}\ T\Rightarrow \tau(\varphi, \entry)$\\
\arraycolsep=2pt
\def\arraystretch{1.4}
\[
\begin{array}{rcl}
\multicolumn{3}{l}{\textbf{Rewrite rules}\ \tau(\varphi, r_0):}\\
\tau({\sf true}, r_0) & \hookrightarrow & {\sf true}\\
\tau(a\in D, r_0) & \hookrightarrow & 
\left\{
\begin{array}{ll}
{\sf true}\quad &\text{if}\ L(r_0)(a)\in D\\
{\sf false} &\text{otherwise}\\
\end{array}\right.\\
\tau(\neg \varphi, r_0) & \hookrightarrow& \neg \tau(\varphi, r_0)\\
\tau(\varphi_1\wedge \varphi_2, r_0) & \hookrightarrow& \tau(\varphi_1, r_0) \wedge \tau(\varphi_2, r_0)\\
\tau({\sf EX}\varphi, r_0) & \hookrightarrow & \exists{r_1\!\in\! E(r_0)}.\ \big( C_{r_0, r_1} \wedge \tau(\varphi, r_1)\big)\\
\tau({\sf AX}\varphi, r_0) & \hookrightarrow & \forall{r_1\!\in\! E(r_0)}.\ \big( C_{r_0, r_1} \Rightarrow \tau(\varphi, r_1)\big)\hspace{20pt}\\
\tau({\sf E}[\varphi_1 {\sf U} \varphi_2], r_0) & \hookrightarrow & \tau_{\sf U}({\sf E}[\varphi_1 {\sf U} \varphi_2], r_0, \emptyset)\\
\tau({\sf A}[\varphi_1 {\sf U} \varphi_2], r_0) & \hookrightarrow & \tau_{\sf U}({\sf A}[\varphi_1 {\sf U} \varphi_2], r_0, \emptyset)
\end{array}
\]
\[
\begin{array}{l}
\textbf{Rewrite rules}\ \tau_{\sf U}(\varphi, r_0, X),\ \text{with}\ X\subseteq {\cal R}:\\
%
%
\tau_{\sf U}({\sf E}[\varphi_1 {\sf U} \varphi_2], r_0, X) \hookrightarrow \tau(\varphi_2, r_0) \vee \Big(\tau(\varphi_1, r_0)\wedge \\ 
\hspace{10pt} 
\big( \exists{r_1\!\in\! E(r_0)\!\setminus\! X}.\ C_{r_0, r_1} \wedge \tau_{\sf U}({\sf E}[\varphi_1 {\sf U} \varphi_2], r_1, X\!\cup\!\{r_0\}) \big) \Big) \\
%
%
\tau_{\sf U}({\sf A}[\varphi_1 {\sf U} \varphi_2], r_0, X) \hookrightarrow \tau(\varphi_2, r_0) \vee \Big(\tau(\varphi_1, r_0) \wedge \\
\hspace{10pt} 
\big( \forall{r_1\!\in\! E(r_0)\!\setminus\! X}.\ C_{r_0,r_1}\!\Rightarrow\!\tau_{\sf U}({\sf A}[\varphi_1 {\sf U} \varphi_2], r_1, X\!\cup\!\{r_0\}) \big) \wedge\\
\hspace{10pt} \big(\forall{r_1\!\in\! E(r_0)\! \cap\! X}.\ \neg C_{r_0,r_1}\big) \Big) \\
\end{array}
\]
\end{minipage}}
\caption{Encoding the satisfaction of a requirement $T\Rightarrow \varphi$ in a resource structure $S = ({\cal R}, E, \entry, L)$, given a template~$C$, into an SMT constraint.
The rewrite rules~$\tau$ reduce an access constraint $\varphi$ and a resource $r_0$ to an SMT constraint. For a resource~$r_0\in\mathcal{R}$, we write $E(r_0)$ for $\{r_1\in \mathcal{R}\mid (r_0,r_1)\in E\}$. The
$\exists$ and $\forall$ quantifiers range over a finite domain. Therefore, the former can be expanded as a finite number of disjunctions, and the latter as a finite number of conjunctions.
 }
\label{fig:encode}
\end{figure}

The conjunction of the constraints created for all the requirements is
called~$\phi$ in Algorithm~\ref{alg:smt}. To check whether there is a
configuration in~$C$ that satisfies the requirements, the
algorithm calls an SMT solver to find a model for the formula $\exists \vec{z}. \forall \vec{a}.\ \phi$. Here this is
\[ \exists \vec{z}. \forall \vec{a}.\ (\psi_{R2}\wedge \psi_{R5})~, \]
where $\vec{z}$ and $\vec{a}$ consist, respectively, of all the
control and attribute variables.
If~$\phi$ is unsatisfiable, then no configuration in $C$ satisfies the
requirements. In this case, the algorithm returns~$\mathsf{unsat}$.
If however the formula is satisfiable, then
the SMT solver returns a model of the formula, which instantiates all
the control variables (but not the attribute variables since they are
universally quantified).
We refer to the SMT solver's procedure that returns such a model as
{\sc Model} in Algorithm~\ref{alg:smt}.  
The model~$\mathcal{M}$ generated by the SMT solver in effect
identifies the local policy for each edge~$e$: by instantiating the
control variables in~$C(e)$, we obtain~$e$'s local policy; see
template~\ref{ex:template}.  This procedure is called $\textsc{Derive}(C(e), {\cal M})$ in the algorithm.
For our example, a model $\cal M$ that satisfies $\exists \vec{z}. \forall \vec{a}.\ (\psi_{R2}\wedge \psi_{R5})$ maps $z_{(\sf cor, bur)}$ to $2$, $z_{(\sf out, cor)}$ to $3$, and all other control variables to $1$.
It is then evident from template~\ref{ex:template} that, e.g., the local policy for the edge $({\sf cor, bur})$ is $({\sf role} = {\sf employee})$.

\para{Complexity}
Let $S$ be a resource structure, $R$ be a set of requirements, and $C$
be configuration template.
The running time of the ${\cal S}_{\sf smt}$ algorithm is determined by
the size of the generated formula $\phi$ and
the complexity of finding a model of~$\phi$.
The size of the formula~$\phi$ is in~${\cal O}(d\cdot |R|\cdot |{\cal R}|)$, where~$d$ is the size of the largest access constraint that appears in the requirements, $R$ is the set of requirements, and $\cal R$ is the set of resources in~$S$.
The formula $\phi$ is defined over Boolean control variables~$\vec{z}$
and attribute variables~$\vec{a}$. The number of control and attribute variables is
$\lceil{\it log}(|C|)\rceil$ and
$|{\cal A}|$, respectively. In the worst case, one must check all possible models of
the formula $\phi$, so finding a model of $\phi$ is in 
${\cal O}(2^{\lceil{\it log}(|C|)\rceil +
k\cdot |A|})$, where $k$ is the largest domain
that appears in the constraints.
Note that such domains are always finite. For example, ${\sf time}
\geq 10$ is a shorthand for $\neg ({\sf time} \in \{0, \ldots, 9\})$.
We conclude that the overall running time of the algorithm ${\cal S}_{\sf smt}$ is in
${\cal O}(2^{\lceil{\it log}(|C|)\rceil + 
k\cdot |A|} + d\cdot |R|\cdot |{\cal R}|)$.

\subsection{Soundness and Completeness}
The algorithm ${\cal S}_{\sf smt}$ is sound. 

\begin{theorem}
Let $S$ be resource structure, $R$ a set of requirements, and $C$ a configuration template. If ${\cal S}_{\sf smt}(S, R, C) = c$ then $S, c\Vdash R$. If ${\cal S}_{\sf smt}(S, R, C) = {\sf unsat}$, then there is no configuration $c$ in $C$ such that $ S, c\Vdash R$.
\label{thm:soundness} 
\end{theorem}

${\cal S}_{\sf smt}$'s completeness depends on the template~$C$ provided as input to the algorithm. We show that one can construct a template for which the algorithm is complete.
A template $C$ is {\em complete} for a given resource structure~$S$
and set of requirements~$R$ if $\SMT(S, R, C)$ returns a configuration
whenever there is a configuration that satisfies the requirements.
For the algorithm's completeness, it is in fact sufficient to start the
algorithm with a template~$C_{S,R}$ that contains all the
configurations that the algorithm based on controller synthesis,
described in Section~\ref{sec:decidability}, may output. 
The following theorem formalizes this  observation.

\begin{theorem}
Given a resource structure~$S$ and a set~$R$ of requirements, the
configuration template~$C_{S,R}$ is complete for $S$ and~$R$.
\label{thm:completeness} 
\end{theorem}

The number of configurations in $C_{S,R}$ is exponential in $|E|$ and
$|R|$ (which we prove in~\cite{tech-report}).  Hence this template, although
complete, is not useful in practice as it would overwhelm SMT solvers,
rendering~$\SMT$ ineffective.  In \secref{sec:implementation}, where we explain our implementation in detail, we describe a configuration template that works
well for synthesizing configurations for practically-relevant
examples.

We conclude this discussion by pointing out that our synthesis
algorithm can be readily used to verify whether a candidate
configuration~$c$ satisfies a set~$R$ of
global access-control requirements in a resource
structure~$S$. Namely, if the configuration template input to $\SMT$
consists only of the configuration~$c$, then $\SMT$ returns~$c$
if~$S,c\models R$; otherwise, the algorithm returns ${\sf unsat}$,
which means that the configuration~$c$ does not satisfy~$R$.

\subsection{Encoding into SMT}
We now explain Algorithm~\ref{alg:smt}'s procedure~{\sc Encode}, which translates a resource structure $S$, a requirement $R = (T\Rightarrow \varphi)$, and a configuration template~$C$, into an SMT constraint $T\Rightarrow \tau(\varphi, \entry)$. The generated constraint encodes that whenever the requirement $T\Rightarrow \varphi$ is applicable to an access request~$q$, i.e. $q\vdash T$, then $\varphi$ must be satisfied for the entry resource $\entry$ in the structure~$S_{c,q}$. Here,~$c$ is the configuration selected from the template~$C$. 
The  constraint $\tau(\varphi, \entry)$ is generated using the rewrite rules $\tau$ as defined in Figure~\ref{fig:encode}. 

Given an access constraint $\varphi$ and a resource~$r_0$, the rewrite rules $\tau$ produce an SMT constraint $\tau(\varphi, r_0)$ that encodes $S, r_0\models \varphi$; see Figure~\ref{fig:semantics}.
The rewrite rules for access constraints of the form $\sf true$, $a\in D$, $\neg \varphi$, and $\varphi_1\wedge \varphi_2$ are as expected.
The rewrite rule for access constraints of the form ${\sf EX}\varphi$
encodes that the access constraint~$\varphi$ is satisfied at~$r_0$ if there is an edge from~$r_0$ to some node~$r_1$ such that $C_{r_0, r_1}$ holds and $S, r_1\models \varphi$. In this rule, the constraint $C_{r_0, r_1}$ returns the symbolic encoding of the local policies for the edge $(r_0,r_1)$, and $\tau(\varphi, r_1)$ returns the encoding of $S, r_1\models \varphi$ as an SMT constraint.
In contrast to~${\sf EX}$, the rewrite rule for ${\sf AX}\varphi$ access constraints states that for any resource~$r_1$, such that~$(r_0,r_1)\in E$, if $C_{r_0,r_1}$ is true then
the constraint~$\tau(\varphi, r_1)$ is satisfied. 

To encode the semantics of the connectives ${\sf EU}$ (${\sf AU}$), we use the \emph{until} 
rewrite rules $\tau_{\sf U}$, which reduce an until construct ${\sf E}[\varphi_1 {\sf U}\varphi_2]$  (${\sf A}[\varphi_1 {\sf U}\varphi_2]$), a resource $r_0\in {\cal R}$, and a set of resources $X\subseteq {\cal R}$ to an SMT constraint. We use the set of resources $X$ to record for which resources the satisfaction of the until access constraint has already been encoded. This is necessary to guarantee  the reduction system's termination.
The rule for access constraints of the form ${\sf E}[\varphi_1 {\sf U}\varphi_2]$ encodes that either $S,r_0\models \varphi_2$, or $S,r_0\models \varphi_1$ and there is an edge from~$r_0$ to some node~$r_1$ such that $C_{r_0, r_1}$ holds and $S, r_1\models {\sf   E}[\varphi_1 {\sf U}\varphi_2]$. Here $\tau_{\sf U}(E[\varphi_1 {\sf U}\varphi_2], r_1, X\cup \{r_0\})$ returns the encoding of $S, r_1\models {\sf E}[\varphi_1 {\sf U}\varphi_2]$. Note that we add
$r_0$ to $X$ to ensure that no resource is revisited during~${\sf   EU}$-rewriting.
Similarly, the rule for access constraints ${\sf A}[\varphi_1 {\sf U}\varphi_2]$ encodes that either $S,r_0\models \varphi_2$, or $S,r_0\models \varphi_1$ holds, for any outgoing edge to a node~$r_1$ we have  $S, r_1\models {\sf A}[\varphi_1 {\sf U}\varphi_2]$
and it has no outgoing edges to nodes in $X$. 
\ifext
We illustrate the encoding with examples in Appendix~\ref{sec:soundness}. There, we also prove that this rewrite system always terminates, and the generated SMT encoding of access constraints is correct.
\else
We illustrate the encoding with examples in~\cite{tech-report}. There, we  also prove that this rewrite system always terminates, and the generated SMT encoding of access constraints is correct.
\fi
\section{Implementation and Evaluation}

\label{sec:eval} 

We report on an implementation of our policy synthesis algorithm, the case studies we conducted to evaluate its efficiency and scalability, and our empirical results.

\subsection{Implementation}

\label{sec:implementation}

We have implemented a synthesizer that encodes policy synthesis instances into the QF\_LIA and QF\_UA logics of SMT-LIB v2~\cite{Barrett10c} and uses the Z3 SMT solver~\cite{DeMoura:2008:ZES:1792734.1792766}.
%
%
%
Our synthesizer is configured with configuration templates of different sizes. 
The local policies defined by these configuration templates are in disjunctive normal form. 
Namely, the local policies are defined as a disjunction of clauses, each clause consisting of a conjunction of terms, where each term is either an equality constraint for non-numerical
attributes (e.g.  ${\sf role} = {\sf employee}$) or an interval constraint for numeric attributes (e.g. $t_1 \le{\sf time}\le t_2$).
We denote by~$C_k$ the configuration template that defines local policies with $k$ clauses, each consisting of $k$ terms.
Note that the local policies defined in the template $C_k$ may refer to at most $k^2$ attributes.

Our synthesizer implements the following procedure: 
it iteratively executes $\SMT(S, R, C_1)$, $\SMT(S, R, C_2)$, $\SMT(S, R, C_3), \ldots$, stopping with the first call to $\SMT$ that returns a satisfying configuration, and returning this configuration.
By iterating over templates increasing in size, our synthesizer generates small local policies, which is desirable for avoiding redundant attribute checks.
For the running example, for instance, our synthesizer's output includes the constraint $\operatorname{\sf correct-pin}$ only for the entrance gates' local policies, and does not include this check, e.g., for the office room's policy. 
A satisfying solution for each case study can be found in the configuration template~$C_3$. This indicates that real-world local policies have concise representations.

Note that our synthesizer may not terminate in a reasonable amount of time if no configuration satisfies the global requirements for the given resource structure. 
In our case studies, we used a simple iterative method to pinpoint such unsatisfiable requirements: we start with a singleton set of requirements, consisting of one satisfiable requirement, and iteratively extend this set by one requirement. This helped us identify a minimal set of conflicting requirements and revise problematic ones.

\subsection{Case Studies} \label{case-studies-sec}
To investigate $\mathcal{S}_\mathsf{smt}$'s efficiency and scalability, we have conducted case studies in collaboration with \kaba. We used real-world requirements and resource structures, and used our tool to synthesize policy configurations for a university building, a corporate building, and an airport terminal.
Our synthesizer and all data are publicly available\footnote{\url{http://www.infsec.ethz.ch/research/software/spctl.html}}. Below, we briefly explain
the three case studies; relevant complexity metrics are summarized in Table~\ref{tab:metrics}.

\para{University Building}
We modeled the main floor of ETH~Zurich's computer science building. This floor consists of $66$ subspaces including labs, offices, meeting rooms, and shared areas. 
The subspaces are labeled with four attributes that indicate: the research group to which a physical space is assigned, the physical space type (e.g., office, teaching room, or server room), the room number, and whether the physical spaces belongs to a secretary or a faculty member. Example requirements stipulate that a research group's PhD~students can access all offices assigned to the group except those assigned to the faculty members and secretaries. The policies are defined over eight attributes. 

\para{Corporate Building}
We modeled an office space that consists of $20$ subspaces, including
a lobby, meeting rooms, offices, and restricted
areas such as a server room, a mail room, and an HR office.  The rooms
are connected by three corridors,
and they are labeled with attributes to mark public areas and employee-only zones.
%
Access to these spaces is controlled by locks that are equipped with smartcard readers and PIN keypads. These locks are connected to a time server. 
Example requirements are that only the postman and HR employees can access the mail room, and that between noon and 1PM employees can access their offices without entering their PIN.  The policies are defined over four attributes.  

\para{Airport Terminal}
We modeled the main terminal of a major international airport. The part of the terminal that we modeled includes subspaces such as the boarding pass control, security, and shopping areas.
We have used the actual plan of the terminal, and considered $15$
requirements, all currently enforced by the airport's access-control
system. The area is divided into $13$ subspaces, each
labeled with zone identifiers (such as check-in and
passport control). Example requirements stipulate that no passenger can
access departure areas before passing through security, passengers with
economy boarding passes cannot pass through the business/first-class
ticket-control gates, and that only airport staff can access certain
elevators.

\begin{table}
	\centering
	\normalsize
	\tabcolsep2pt
	\begin{tabular}{l|lrrr}
		\toprule
		\multicolumn{2}{l}{} & \multicolumn{1}{l}{University}& \multicolumn{1}{l}{Corporate} & \multicolumn{1}{l}{Airport}\\
		\multicolumn{2}{l}{} & \multicolumn{1}{l}{building}& \multicolumn{1}{l}{building}& \multicolumn{1}{l}{terminal}\\
		\midrule
		Complexity & Requirements& 14  & 10 & 15\\
		metrics & PEPs & 127 & 41 & 32\\
		& Subspaces & 66 & 20 & 13\\
		\midrule
		Performance & Synthesis time & 10.32 & 25.30 & 1.92\\
		& Std. dev. & 0.04 & 0.15 & 0.01\\
		\bottomrule
	\end{tabular}
	\caption{Complexity metrics and policy synthesis times (in seconds) for the three cases studies}
	\label{tab:metrics}
\end{table}

\subsection{Empirical Results}
We ran all experiments on a Linux machine with a quad-code
$\operatorname{i7-4770}$ CPU, $32$GB of RAM, running Z3 SMT~v$4.4.0$.
We present two sets of results: (1) the
synthesizer's performance when used to synthesize the local policies
for the three case studies, and 
(2) the synthesizer's  scalability.

\para{Performance}
We used our tool to synthesize the local policies for the three
case studies, measuring the time taken for policy
synthesis.  We report the average synthesis time, measured over $10$
runs of the synthesizer, in the bottom two rows of Table~\ref{tab:metrics}.  The reported
synthesis time is the sum of the time taken for encoding the policy
synthesis instance into SMT constraints, the time for
solving the generated SMT constraints, and the time for iterating over
the smaller templates for which the synthesizer returns $\sf unsat$.
In all three case studies, our tool synthesizes the local
policies in less than $30$ seconds.
The standard deviation is under $0.2$ seconds.
This indicates that synthesizing local policies is practical, and can
be used for real-world systems.

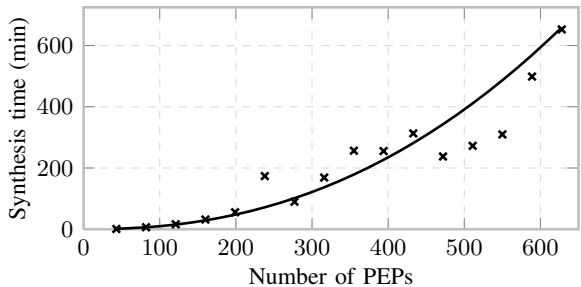
\begin{figure}
\centering
\begin{tikzpicture}
  \tikzstyle{every node}=[font=\small]
  \begin{axis}
  [
    width=0.45\textwidth,
    height=0.25\textwidth,
    grid=major,
    grid style={dashed,gray!30},
    xlabel={Number of PEPs},
    legend cell align=left,
    ylabel={Synthesis time (min)},
    x label style={at={(axis description cs:0.5,0.05)},anchor=north},
    y label style={at={(axis description cs:0.1,0.5)},anchor=south},
    draw=gray!50,
    line width=1pt,
    xmin=0,
    ymin=0,
    xmax=650,
    scaled y ticks = false,
    scaled x ticks = false,
    legend style={at={(0,1)},anchor=north west},
    xticklabel = \pgfmathprintnumber{\tick},
  ]
  \addplot+[only marks, mark=x, mark options={fill=black},color=black,solid,line width=1pt] table[x=peps,y=time] {data/scalability.txt}; 
  \addplot+[smooth, mark=none,color=black,solid,line width=1pt] table[x=peps,y=trend] {data/scalability.txt}; 
  \end{axis}
\end{tikzpicture}
\caption{Scaling the number of PEPs}
\label{fig:scale-res}
\end{figure}

\para{Scalability Experiments}
To investigate the scalability of our synthesis tool, we synthetically generated larger problem instances based on the corporate building case study.  Although the case study originally consisted of a single floor, we increased the number of the floors in the building.
We kept the same labeling for the newly added subspaces, so the original requirements also pertain to the newly added floors.
Based on this method, we scaled the number of PEPs up to $650$.

The time needed to synthesize local policies for different numbers of PEPs is given in Figure~\ref{fig:scale-res}. The results show that our tool can synthesize a large number of local policies in a reasonable amount of time. For example, synthesizing up to $500$ local policies  takes less than five hours. 
The tool's performance can be further improved using domain-specific heuristics for solving the resulting SMT constraints.
Nevertheless, the tool already scales to most real-world scenarios: protected physical spaces usually have less than $500$ PEPs.
\section{Related Work}

\label{sec:rw}

\para{Physical Access Control}
The Grey project was an experiment in deploying a physical
access-control system at the campus of Carnegie Mellon
University~\cite{bauer:distprove,bgr07}. As part of this project,
researchers developed formal languages for specifying policies and
credentials, and also developed techniques for detecting policy
misconfigurations~\cite{prediction:codaspy12,prediction-tissec}.
The work on credential management, such as delegation, is orthogonal
to the specification of the locks' local policies. In contrast to their
work on detecting policy misconfigurations, 
we have developed a framework to synthesize policies
that are guaranteed to 
enforces the global requirements, avoiding misconfigurations.

Several researchers have investigated SAT-based and model-checking
techniques for reasoning about physical access
control~\cite{Fitzgerald_anomalyanalysis,vmcai11}.  Similarly to our
work, these approaches model spatial constraints, and formalize global
requirements that physical access-control systems must enforce.  The
authors of~\cite{Fitzgerald_anomalyanalysis}, for instance, model
physical spaces using directed graphs and formalize global
requirements in first-order logic. Their goal is to identify undesired
denials due to blocked paths and unintended grants to restricted zones
using SAT solvers. In contrast to these verification approaches, we
develop a synthesis framework for generating correct local policies.

\para{Network Policy Synthesis}
The problems of configuring networks with access-control and routing
policies are related to the problem of constructing local policies
from global requirements.
In the network problem domain, one has an explicit resource structure defined
by the network topology and must enforce global requirements using
local rules deployed at the switches. 
Several synthesis algorithms for networks have been studied;
e.g.\ see~\cite{DBLP:conf/sp/Guttman97,firmato99,McClurg:2015:ESN:2737924.2737980,Narain2008,Zhang:2011:SDF:2147671.2147677,RahmanA13,mooly:popl15}.
The authors of \cite{DBLP:conf/sp/Guttman97} and \cite{firmato99}, for
example, propose techniques for synthesizing local firewall rules that
collectively enforce global network requirements in a given network
topology. These approaches are sufficiently expressive for formalizing
simple connectivity constraints, such as which hosts can access which
services in a network.
Similarly to our approach, recent techniques for synthesizing network
configurations, such
as~\cite{McClurg:2015:ESN:2737924.2737980,Narain2008,Zhang:2011:SDF:2147671.2147677,RahmanA13},
also leverage SAT and SMT solvers. In addition to access-control
constraints, these techniques also consider business constraints, such
as deployment cost and usability.
However, none of the above approaches for network synthesis supports
branching properties, which are necessary for specifying requirements
such as those stipulating that a fire-exit is reachable from any
office room, as well as those that instantiate our waypointing and
blocking requirement patterns; see Section~\ref{sec:patterns} for
examples.  These requirements, which can be expressed in our
framework, are central to physical access control. Existing
network policy synthesis algorithms, therefore, are not sufficiently
expressive for handling access-control requirements for physical spaces.

Policy verification has also been studied in the context of computer
networks; see e.g.~\cite{infocom1354680}. However, this line of research is not
concerned with synthesis, which is our work's main focus. We remark though that our synthesis algorithm can be readily
used for verifying the conformance of a set of local policies to
global access-control requirements; see Section~\ref{sec:algorithm}.

\para{Program Synthesis}
Program synthesis techniques, such as template-based
synthesis~\cite{armando-phd, Srivastava:2010:PVP:1706299.1706337,
  armando-fmcad13,Srivastava:2011:PIS:1993498.1993557}, reactive
program synthesis from temporal
specifications~\cite{Clarke:1981:DSS:648063.747438,jobstmann06,
  Morgenstern:2011:PSV:2032692.2032706}, and program repair
techniques~\cite{Jobstmann:2005:PRG:2153230.2153260,
  Buccafurri:1999:EMC:319103.319105}, are related to policy synthesis
for physical spaces.
Similarly to our SMT-based algorithm,
most of these synthesis frameworks also supplement the logical
specification with a template, and  exploit SMT solvers to
efficiently explore the search space defined by the template.
They cannot however express the relevant access-control requirements
we have considered, such as those  pertaining to branching
properties.
Our synthesis framework builds upon these techniques, and extends them
with support for specifications that are needed for physical spaces.

 Methods for synthesizing models of logical formulas, such as those in
 linear-temporal logic or CTL, have been extensively studied in the
 literature~\cite{Clarke:1981:DSS:648063.747438,Pnueli:1989:SRM:75277.75293,Ramadge:1987:SCC:35469.35482,pistore-tr,pistore-mbp,
   antoniotti95}. 
In Section~\ref{sec:synthesis}, we have described a policy
synthesis algorithm based on CTL controller synthesis. This algorithm however comes
at the expense of an exponential blow-up.
Therefore, existing CTL synthesis tools and algorithms cannot be readily applied to synthesize attribute-based local policies in practice.
Our efficient SMT-based algorithm addresses this practical challenge.
\section{Conclusion}
\label{sec:conc}

We have presented a framework for synthesizing locally enforceable
policies from global, system-wide access-control requirements for
physical spaces.  Its key components are (1)~a~declarative language
along with patterns for writing global requirements, (2) a model of
the physical space describing how subjects access resources, and (3)
an efficient policy synthesis algorithm for generating policies
compliant with the requirements and the spatial constraints.
Using real-world case studies, we have demonstrated that our synthesis framework is practical and  scales to systems with complex requirements and numerous policy enforcement points.

As future work, we plan to extend our policy synthesis framework with architectural constraints. Examples include \emph{optimality} constraints, which can be used to synthesize local policies that avoid re-checking attributes that have been checked by other enforcement points.
Handling such constraints is important for large-scale access-control systems in practice.
We also plan to apply our framework to synthesize locally enforceable policies in other access-control domains, such as access control in networks, which may require tailored synthesis heuristics.

\bibliographystyle{IEEEtran}
\bibliography{bib}

\ifext
\appendix
\subsection{Correctness and Complexity of the Algorithm~$\CS$}

\label{sec:decidability}

\para{Correctness}
We now prove the correctness of~${\cal S}_{\sf cs}$.

\begin{theorem}
Let $S$ be a resource structure and $R$ a set of requirements.
If ${\cal S}_{\sf cs}(S, R) = c$ then $S,c\Vdash R$. 
If ${\cal S}_{\sf cs}(S, R) = {\sf unsat}$ then there is no configuration $c$ such that $S, c\Vdash R$.
\label{thm:decompose} 
\end{theorem}

\begin{proof}
We prove the two implications by contradiction.

Assume that ${\cal S}_{\sf cs}(S, R)$ returns a configuration~$c$. 
Suppose for the sake of contradiction that $S,c\not\Vdash R$. 
Then, by definition of $\Vdash$, there is an access request $q$ and a requirement $T\Rightarrow \varphi$ such that $q\vdash T$ and $S_{c,q}\not\models \varphi$. 
Given a subset $R'\subseteq R$ of the requirements, let $T_{R'} = T_1\wedge ...\wedge T_i \wedge \neg T_{i+1} \wedge \cdots \wedge \neg T_n$, where $\{T_1\Rightarrow \varphi_1, \ldots, T_i\Rightarrow \varphi_i\} = R'$ and $\{T_{i+1}\Rightarrow \varphi_{i+1}, \ldots, T_n\Rightarrow \varphi_n\} = R\setminus R'$.
The constraint $T_{R'}$ corresponds to the target computed at line~\ref{alg:cs:target} of Algorithm~\ref{alg:cs}.
Let $R_q = \{(T\Rightarrow \varphi)\in R\mid q\vdash T\}$ be the set of all requirements in $R$ that are applicable to~$q$.
We have $q\vdash T_{R_q}$ (1). Furthermore, for any $R'\subseteq R$ where $R'\neq R_q$, we have $q\not\vdash T_{R'}$ (2).
By definition of $S_{c,q}$, $S_{c,q}$ contains an edge $e$ if $e$ is an edge of $S$ and $q\vdash c(e)$. Algorithm~\ref{alg:cs} constructs the configuration $c$ by conjoining targets $T_{R'}$, where $R'\subseteq {\cal R}$, to the local policies $c(e)$; see line~\ref{alg:cs:conf}. From (1) and (2) we conclude the following: 
First, adding $\neg T_{R_q}$ to a local policy $c(e)$ removes the edge $e$ in $S_{c,q}$ because $q\not\vdash c(e)\wedge (\neg T_{R_q})$. Second, adding $\neg T_{R'}$ to a local policy $c(e)$, where $R'\neq R_q$, does not remove the edge $e$ in $S_{c,q}$ because $q\vdash c(e) \wedge (\neg T_{R'})$ iff $q\vdash c(e)$. 
It is immediate that $S_{c,q}$ contains those edges of $S$ for which the target $\neg T_{R_q}$ is not conjoined to the local policy $c(e)$. We conclude that $S_{c,q}$ contains the edges $E' = {\sf cs}(S, \varphi_{R_q})$ (see line~\ref{alg:cs:edges} of Algorithm~\ref{alg:cs}), where $\varphi_{R_q}$ conjoins the access constraints of all requirements in $R_q$.
By definition of controller synthesis, we have $(S, E', r, L)\models \varphi_{R_q}$. Since $S_{c,q} = ({\cal R}, E', r, L)$, $S_{c,q}\models \varphi_{R_q}$.
We can now deduce that $S_{c,q}\models \varphi$ because $(T\Rightarrow\varphi)\in R_q$. But previously we deduced that $S_{c,q}\not\models \varphi$. Thus we have a contradiction, and there is no access request $q$ and requirement $T\Rightarrow \varphi$ such that $q\vdash T$ and $S_{c,q}\not\models \varphi$. Therefore, $S, c\Vdash  R$.

Assume that ${\cal S}_{\sf cs}(S, R) = {\sf unsat}$. Suppose for the sake of contradiction that there is a configuration $c$ such that $S, c\Vdash R$.
From ${\cal S}_{\sf cs}(S, R) = {\sf unsat}$, by definition of Algorithm~\ref{alg:cs}, it follows that there is a subset $R' = \{T_1\Rightarrow \varphi_1, \ldots, T_k\Rightarrow \varphi_k\}\subseteq R$ of the requirements and an access request $q$, such that $q\vdash T_1\wedge \cdots\wedge T_k$ (1) and ${\sf cs}(S, \varphi_1\wedge \cdots\wedge \varphi_k) = {\sf unsat}$ (2).
From (1), we know that all requirements in $R'$ are applicable to $q$ .
Furthermore, since $S, c\Vdash R$, it must be that $S_{c,q}\models \varphi_i$, for $1\leq i\leq k$. We get $S_{c,q}\models \varphi_1\wedge \cdots\wedge \varphi_k$.
From (2), by definition of controller synthesis, there is no resource structure $S' = ({\cal R}, E', r, L)$, with $E'\subseteq E$, such that $S'\models \varphi_1\wedge \cdots\wedge \varphi_k$. Thus we have a contradiction, and we conclude that there is no configuration $c$ such that $S,c\Vdash R$.

This concludes our proof.
\end{proof}

\para{Complexity}
The running time of algorithm ${\cal S}_{\sf cs}$ is determined by the
number of iterations of the loops, the complexity of checking
the satisfiability of the conjunction of targets (line~\ref{line:cs-sat}),
and the complexity of solving each controller synthesis instance
(line~\ref{line:cs-synth}).
The first loop is executed $|E|$ times, where $|E|$ is the number of
edges, and the second loop is executed $2^{|R|}$ times. The second
loop checks one satisfiability instance and one controller synthesis
instance. The complexity of checking satisfiability is in ${\cal
	O}(2^{k\cdot |A|})$ where $k = |D_{\sf max}|$ for the largest set
$D_{\sf max}$ of values that appears in the constraint $T$, and $|A|$ is the
number of attributes.
Solving a controller synthesis instance requires checking $({\cal R}, E', \entry, L) \models \varphi$ at most $2^{|E|}$ times, where $E'\subseteq E$ and $\varphi$ is a conjunction of access constraints. 
The problem $({\cal R}, E', \entry, L) \models \varphi$ can be decided using the model checking algorithm for CTL based on labeling, which is in ${\cal O}( |\varphi| \cdot (|{\cal R}|+ |E'|))$, where  $|\varphi|$ is the size of the access constraint $\varphi$~\cite{huth11}. The size of the largest access constraint given as input to $\sf cs$ is in ${\cal O}(|d|\cdot |R|)$, where $d$ is the largest access constraint that appears in the requirements $R$.
The running time of ${\cal S}_{\sf cs}$ is therefore ${\cal O}(2^{|R|}\cdot (2^{k\cdot |A|}  + 2^{|E|}\cdot |R|\cdot d \cdot (|{\cal R}|+ |E|)))$.

\subsection{SMT Encoding}

\label{sec:soundness}

\para{Example}
We illustrate the SMT encoding of an exists-until and an always-until access constraint in Figure~\ref{fig:until-constraints}. The SMT encoding of the access constraint ${\sf E}[(\neg {\sf sec\_zone}){\sf U}({\sf id}={\sf bur})]$ for the resource~$\sf out$ formalizes that the two PEPs $({\sf out, cor})$ and $({\sf cor, bur})$ grant access or the PEP $({\sf out, bur})$ grants access. This guarantees the existence of a path that satisfies the access constraint.
The SMT encoding of ${\sf A}[(\neg {\sf sec\_zone}){\sf U}({\sf id}={\sf bur})]$ for the resource $\sf out$ formalizes that the always-until constraint is satisfied along any path that starts from the resource $\sf out$. Since any path that start with $\sf (out, bur, \ldots)$ satisfies the access constraint, the SMT constraint imposes no constrains on the PEP $(\sf out, bur)$. However, not all paths that start with $(\sf out, cor, \ldots)$ satisfy the access constraint. Concretely, the infinite path $(\sf out, cor, out, cor, \ldots)$ violates the access constraint. The SMT constraint therefore formalizes that if there are paths starting with $(\sf out, cor, \ldots)$, i.e. the PEP $(\sf out, cor)$ grants access, then the PEP $\sf (cor, out)$ denies access. This guarantees that the path violating the access constraint is not present in the resulting resource structure. Note that, since we consider only deadlock-free resource structures, the absence of the edge $\sf (cor, out)$ guarantees that the resulting resource structure has the edge $\sf (cor, bur)$, and therefore all paths starting with $(\sf out, cor, \ldots)$ continue along resource $\sf bur$.

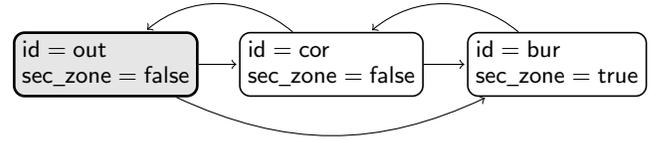
\begin{figure}
\centering
\begin{tikzpicture}[->,shorten >=1pt,auto]
\small
  \node[entry, minimum height=12pt, minimum width=12pt] (out) at (0,0) 
  {\begin{tabular}{l}
  	${\sf id} = {\sf out}$\\
  	${\sf sec\_zone} = {\sf false}$
  \end{tabular}};
  \node[resource, minimum height=12pt, minimum width=12pt] (cor) at ($(out) + (3, 0)$)
  {\begin{tabular}{l}
  	${\sf id} = {\sf cor}$\\
  	${\sf sec\_zone} = {\sf false}$
  \end{tabular}};
  \node[resource, minimum height=12pt, minimum width=12pt] (bur) at ($(cor) + (3, 0)$)
  {\begin{tabular}{l}
  	${\sf id} = {\sf bur}$\\
  	${\sf sec\_zone} = {\sf true}$
  \end{tabular}};  
  \draw[->] (out)-- (cor);
  \draw[->] (cor)-- (bur);
  \draw[->] (out) to [bend right=25] (bur);
  \draw[->] (cor) to [bend right=40] (out);
  \draw[->] (bur) to [bend right=40] (cor);
\end{tikzpicture}
\[
\begin{array}{l}
\tau_{\sf U}({\sf E}[(\neg {\sf sec\_zone}){\sf U}({\sf id} = {\sf bur})], {\sf out}, \emptyset) \hookrightarrow\\[2pt]
\multicolumn{1}{r}{\hspace{130pt}(C_{\sf out,cor}\wedge C_{\sf cor,bur}) \vee C_{\sf out,bur}}\\[6pt]
\tau_{\sf U}({\sf A}[(\neg {\sf sec\_zone}){\sf U}({\sf id}={\sf bur})], {\sf out}, \emptyset) \hookrightarrow\\[2pt]
\multicolumn{1}{r}{(C_{\sf out,cor}\Rightarrow (
\neg C_{\sf cor,out}))
}
\end{array}
\]
\caption{Encoding exists-until and always-until access constraints using SMT constraints.}
\label{fig:until-constraints}
\end{figure}



\para{Termination}
We first prove that the rewrite rules given in Figure~\ref{fig:encode} terminate.

\begin{theorem}
Let $S = ({\cal R}, E, r, L)$ be a resource structure.
For any resource~$r_0\in {\cal R}$ and access constraint $\varphi$, the rewrite function $\tau(\varphi, r_0)$ terminates.
\label{lemma:termination}
\end{theorem}

\begin{proof}
The proof proceeds by induction on the length of the access constrain~$\varphi$. Formally,
we define the {\em length} of an access constraint $\varphi$, denoted by $l(\varphi)$, as 
\[
\begin{array}{rcl}
l({\sf true}) & = & 1\\
l(a\in D) & = & 1\\
l(\neg \varphi) & = & 1 + l(\varphi)\\
l({\sf EX}\varphi) & = & 1 + l(\varphi)\\
l({\sf AX}\varphi) & = & 1 + l(\varphi)\\
l(\varphi_1\wedge \varphi_2) & = & 1 + {\sf max}(l(\varphi_1), l(\varphi_2))\\
l({\sf E}[\varphi_1 {\sf U}\varphi_2]) & = & 1 + {\sf max}(l(\varphi_1), l(\varphi_2))\\
l({\sf A}[\varphi_1 {\sf U}\varphi_2]) & = & 1 + {\sf max}(l(\varphi_1), l(\varphi_2))
\end{array}
\]
where ${\sf max}(n_1, n_2)$ returns $n_1$ if $n_1\geq n_2$, otherwise it returns~$n_2$. Note that $l(\varphi)\geq 1$ for any access constraint~$\varphi$.

\para{Base Case}
For the base case, $l(\varphi) = 1$, the access constraint is of the form $\sf true$ or $a\in D$. The rewrite function $\tau$ terminates in one step. 

\para{Inductive Step}
Assume that $\tau(\varphi, r_0)$ terminates for any access constraint $\varphi$ of length $l(\varphi)\leq k$~(H1). We prove that $\tau(\varphi, r_0)$ terminates for any access constraint of length $l(\varphi) = k+1$.

\begin{compactitem}
\item For the cases where the access constraint~$\varphi$ is of the form $\neg \varphi_1$, ${\sf EX}\varphi_1$, ${\sf AX}\varphi_1$, the rewrite function $\tau(\varphi, r_0)$ calls $\tau(\varphi_1, r_0)$. By induction, $\tau(\varphi_1, r_0)$ terminates because $l(\varphi_1) = k$. 

\item The case where $\varphi = \varphi_1\wedge \varphi_2$ also terminates because $l(\varphi_1)\leq k$ and $l(\varphi_2)\leq k$.

\item For the cases where $\varphi$ is of the form ${\sf E}[\varphi_1 {\sf U} \varphi_2]$ or ${\sf A}[\varphi_1 {\sf U} \varphi_2]$, we need to show that $\tau_{\sf U}(\varphi, r_0, \emptyset)$ terminates. We prove that $\tau_{\sf U}(\varphi, r_0, X)$ terminates for any set $X\subseteq {\cal R}$ by descending induction on the size of the set $X$.
For the base case, we have $|X| = |{\cal R}|$. Then, $\tau_{\sf U}(\varphi, r_0, {\cal R})$ 
calls $\tau(\varphi_1, r_0)$ and $\tau(\varphi_2, r_0)$. By our inductive hypothesis~(H1), both $\tau(\varphi_1, r_0)$ and $\tau(\varphi_2, r_0)$ terminate since $l(\varphi_1)\leq k$ and $l(\varphi_2)\leq k$.
For the inductive step, assume that $\tau_{\sf U}({\sf E}[\varphi_1 {\sf U} \varphi_2]), r_0, X)$ terminates for any $X\subseteq {\cal R}$ of size $k\leq |X|\leq |{\cal R}|$~(H2). Consider a set $X'\subseteq {\cal R}$ of size $|X'| = k-1$. Then, $\tau_{\sf U}(\varphi, r_0, X)$ calls the rewrite functions $\tau(\varphi_1, r_0)$, $\tau(\varphi_2, r_0)$, and $\tau_{\sf U}(\varphi, r_1, X'\cup\{r_0\})$, for $r_1\in E(r_0)\setminus X$. The rewrite function $\tau(\varphi_1, r_0)$, $\tau(\varphi_2, r_0)$ terminate by the inductive hypothesis (H1). By the inductive hypothesis~(H2), the rewrite function $\tau_{\sf U}(\varphi, r_1, X'\cup\{r_0\})$ terminates because $|X\cup \{r_0\}| = k$.
\end{compactitem}

This completes our proof.
\end{proof}

\para{Correctness}
We now prove that the correctness of our SMT-based policy synthesis algorithm. 
We start with several definitions. Our definitions are similar to those used to describe the decision procedure for CTL satisfiability given in~\cite{Emerson:1991:TML:114891.114907}.
Let $\varphi_1$ and $\varphi_2$ be two access constraints and $S = ({\cal R}, E, r, L)$ be a resource structure. 
We assume that $S$ does not contain deadlock resources, i.e.\ for any resource $r_0\in {\cal R}$, the set $E(r_0) = \{r_1\in {\cal R}\mid (r_0, r_1)\in E\}$ is nonempty.
We call access constraints of the form ${\sf A}[\varphi_1{\sf U}\varphi_2]$ and ${\sf E}[\varphi_1{\sf U}\varphi_2]$ eventuality constraints.
We first define the derivation of a rooted directed graph from $S$ for a given access constraint~$\varphi_2$ and root node~$r_0\in {\cal R}$. We call this graph an {\em eventuality graph}. We then give two conditions
over such eventuality graphs. The first condition is satisfied iff $S, r_0\models {\sf A}[\varphi_1{\sf U}\varphi_2]$, while the second one is satisfied iff $S, r_0\models {\sf E}[\varphi_1{\sf U}\varphi_2]$.

We define the eventuality graph $G(S, r_0, \varphi_2)$ as the rooted directed graph obtained by taking the node $r_0$ and all nodes and edges along all paths emanating from $r_0$ up to and including the first node $r_1$ such that $S, r_1\models \varphi_2$; if there is no such node $r_1$ along a path, then all nodes and edges along the path are included in $G(S, r_0, \varphi_2)$. 
We call a node of $G(S, r_0, \varphi_2)$ an {\em interior node} if it has successors; otherwise, we call it a {\em frontier node}.

We now define the two conditions.
We say that an eventuality graph $G(S, r_0, \varphi_2)$ {\em fulfills ${\sf A}[\varphi_1 {\sf U}\varphi_2]$} if 
\begin{compactenum}
\item the graph is acyclic,
\item for any of its interior nodes $r_1$ we have $S, r_1\models \varphi_1$, and
\item for any of its frontier nodes $r_2$ we have $S, r_2\models \varphi_2$.
\end{compactenum}
Note that for resource structures without deadlock resources, (1) implies (3). 
We say that an eventuality graph $G(S, r_0, \varphi_2)$ {\em fulfills ${\sf E}[\varphi_1 {\sf U}\varphi_2]$} if 
\begin{compactenum}
\item the graph contains a frontier node $r_2$ such that $S, r_2\models \varphi_2$, and 
\item there is a path from $r_0$ to this frontier node $r_2$ such that for any interior node $r_1$ along the path we have $S, r_1\models\varphi_1$.
\end{compactenum}
From the CTL satisfiability decision procedure of~\cite{Emerson:1991:TML:114891.114907}, it follows that $S, r_0\models {\sf A}[\varphi_1 {\sf U}\varphi_2]$ iff $G(S, r_0, \varphi_2)$ fulfills ${\sf A}[\varphi_1 {\sf U}\varphi_2]$, and $S, r_0\models {\sf E}[\varphi_1 {\sf U}\varphi_2]$ iff $G(S, r_0, \varphi_2)$ fulfills ${\sf E}[\varphi_1 {\sf U}\varphi_2]$.

To prove the correctness of our SMT-based synthesis algorithm, we first prove that $\tau$ correctly encodes access constraint into SMT constraints. Towards this end, Theorem~\ref{lemma:smt-encoding} establishes that the SMT encoding is correct for any access constraint and any singleton configuration template~$C=\{c\}$, i.e. a template consisting of one configuration. 
To prove this theorem, we give two lemmas (Lemma~\ref{lemma:au} and Lemma~\ref{lemma:eu}), which show that the rewrite function~$\tau_{\sf U}$ correctly encodes eventuality access constraints. 
Afterwards, with Lemma~\ref{lemma:smt-correct-reqs} we lift the correctness of the access constraints' encoding to requirements.
Finally, we restate and prove Theorem~\ref{thm:soundness}.

\begin{theorem}
Let $S = ({\cal R}, E, r, S)$ be a resource structure. For any configuration~$c$ for~$S$, resource $r_0\in {\cal R}$, access request $q\in{\cal Q}$, and access constraint $\varphi$, 
we have 
\[ S_{c,q}, r_0\models \varphi\ \text{iff}\ q\vdash \tau(\varphi, r_0).\]
\label{lemma:smt-encoding}
\end{theorem}

\begin{proof}
%
%
The proof proceeds by induction on the derivation of $\tau(\varphi, r_0)$. 

\begin{itemize}
%
%
\item For the case $\varphi = {\sf true}$, we have $\tau(\varphi, r_0) = {\sf true}$. We get $S_{c,q}, r_0\models \varphi$ and $q\vdash \tau(\varphi, r_0)$.
%
%
\item For the case $\varphi = (a\in D)$, we have $\tau(\varphi, r_0) = {\sf true}$ if  $L(r_0)(a)\in D$, and $\tau(\varphi, r_0) = {\sf false}$ if $L(r_0)(a)\not\in D$. Recall that $S_{c,q}, r_0\models (a\in D)$ iff $L(r_0)(a)\in D$. It is immediate that $S_{c,q}, r_0\models \varphi$ iff $q\vdash \tau(\varphi, r_0)$.
%
%
\item For the case $\varphi = \neg \varphi'$, we have $\tau(\varphi, r_0) = \neg \tau(\varphi', r_0)$. 
\begin{itemize}
\item[$\Rightarrow$:] Assume $S_{c,q}, r_0\models \varphi$. We get $S_{c,q}, r_0\not\models \varphi'$. By induction, $q\not\vdash \tau(\varphi', r_0)$. Therefore $q\vdash \tau(\varphi, r_0)$.
\item[$\Leftarrow$:] Assume $q\vdash \tau(\varphi, r_0)$. We get $q\not\vdash \tau(\varphi', r_0)$. By induction, $S_{c,q}\not\models \varphi'$. Therefore $S_{c,q}\models \varphi$.
\end{itemize}
%
%
\item For the case $\varphi = \varphi_1\wedge \varphi_2$, we have $\tau(\varphi, r_0) = \tau(\varphi_1, r_0)\wedge \tau(\varphi_2, r_0)$. 
\begin{itemize}
\item[$\Rightarrow$:] Assume $S_{c,q}, r_0\models \varphi$. Therefore $S_{c,q}, r_0\models \varphi_1$ and $S_{c,q}, r_0\models\varphi_2$. By induction, $q\vdash \tau(\varphi_1, r_0)$ and $q\vdash \tau(\varphi_2, r_0)$, and therefore $q\vdash \tau(\varphi, r_0)$.
\item[$\Leftarrow$:] Assume $q\vdash \tau(\varphi, r_0)$. Then $q\vdash \tau(\varphi_1, r_0)$ and $q\vdash \tau(\varphi_2, r_0)$. By induction, $S_{c,q}, r_0\models \varphi_1$ and $S_{c,q}, r_0\models \varphi_2$, and therefore $S_{c,q}, r_0\models \varphi$.
\end{itemize}
%
%
\item For the case $\varphi = {\sf EX}\varphi'$, we have $\tau(\varphi, r_0) = \exists{r_1\in E(r_0)}.\ (C_{r_0, r_1} \wedge \tau(\varphi', r_1))$. 
\begin{itemize}
\item[$\Rightarrow$:] Assume $S_{c,q}, r_0\models \varphi$. By definition of $S_{c,q}$, there is an edge $(r_0, r_1)$ in $E$ such that $q\vdash c((r_0, r_1))$ (1) and $S_{c,q}, r_1\models \varphi'$ (2). 
Since $C = \{c\}$, we have $C_{r_0, r_1} = c((r_0, r_1))$. From (1), we thus get $q\vdash C_{r_0, r_1}$.
From (2), by induction, we get $q\vdash \tau(\varphi', r_1)$. It follows that $q\vdash \tau(\varphi, r_0)$.
\item[$\Leftarrow$:] Assume $q\vdash \tau(\varphi, r_0)$. There is an edge $r_1\in E(r_0)$ such that $q\vdash C_{r_0, r_1}$ (1) and $q\vdash \tau(\varphi', r_1)$ (2). From (1), we get $q\vdash c((r_0, r_1))$, and thus there is an edge $(r_0, r_1)$ also in $S_{c,q}$. From (2), by induction, we get $S, r_1\models \varphi'$. Therefore, $S_{c,q}, r_0\models \varphi$.
\end{itemize}
%
%
\item For the case $\varphi = {\sf AX}\varphi'$, we have $\tau(\varphi, r_0) = \forall{r_1\in E(r_0)}.\big( C_{r_0, r_1}\Rightarrow \tau(\varphi', r_1)\big)$. 
\begin{itemize}
\item[$\Rightarrow$:] Assume $S_{c,q}, r_0\models \varphi$. Then, for any edge $(r_0, r_1)$ of $S_{c,q}$ we have $S_{c,q}, r_1\models \varphi'$. 
Consider an edge $(r_0, r_1)\in E$ such that $q\vdash C_{r_0, r_1}$. From $q\vdash C_{r_0, r_1}$, we know that $(r_0, r_1)$ is also an edge in $S_{c,q}$. Therefore, $S_{c,q}, r_1\models \varphi'$. By induction, $q\vdash \tau(\varphi', r_1)$. We get $q\vdash \tau(\varphi, r_0)$.
\item[$\Leftarrow$:] Assume $q\vdash \tau(\varphi, r_0)$. Then for any edge $(r_0, r_1)\in E$, $q\vdash C_{r_0, r_1}$ implies $q\vdash \tau(\varphi', r_1)$. Consider an edge $(r_0, r_1)$ of $S_{c,q}$. We know that $q\vdash C_{r_0, r_1}$, and thus $q\vdash \tau(\varphi', r_1)$. By induction, $S_{c,q}, r_1\models \varphi'$. Therefore $S_{c,q}, r_0\models \varphi$.
\end{itemize}


\item For the case $\varphi = {\sf E}[\varphi_1 {\sf U}\varphi_2]$, we have $\tau(\varphi, r_0) = \tau_{\sf U}(\varphi, r_0, \emptyset)$. 
By induction, for any resource $r_1\in {\cal R}$ and any access request $q\in {\cal Q}$ we have
\[
\bigwedge_{i\in \{1,2\}} S_{c,q}, r_1\models \varphi_i\ \text{iff}\ q\vdash \tau(\varphi_i, r_1).
\]
%
By Lemma~\ref{lemma:eu}, we get $S_{c,q}, r_0\models {\sf E}[\varphi_1 {\sf U}\varphi_2]$ iff $q\vdash \tau_{\sf U}(\varphi, r_0, \emptyset)$.

\item For case $\varphi = {\sf A}[\varphi_1 {\sf U}\varphi_2]$, have $\tau(\varphi, r_0) = \tau_{\sf U}(\varphi, r_0, \emptyset)$. 
By induction, for any resource $r_1\in {\cal R}$ and any access request $q\in {\cal Q}$ we have
\[
\bigwedge_{i\in \{1,2\}} S_{c,q}, r_1\models \varphi_i\ \text{iff}\ q\vdash \tau(\varphi_i, r_1).
\]
By Lemma~\ref{lemma:au}, we get $S_{c,q}, r_0\models {\sf A}[\varphi_1 {\sf U}\varphi_2]$ iff $q\vdash \tau_{\sf U}(\varphi, r_0, \emptyset)$.

\end{itemize}

This concludes our proof.
\end{proof}

\begin{lemma}
Let $S = ({\cal R}, E, r, S)$ be a resource structure, $\varphi_1$ and $\varphi_2$ be two access constraints, and $C = \{c\}$ be a configuration template. 
If for any resource $r_1\in {\cal R}$ and any access request $q\in {\cal Q}$ we have
\begin{equation}\tag{A1}
\bigwedge_{i\in \{1,2\}} S_{c,q}, r_1\models \varphi_i\ \text{iff}\ q\vdash \tau(\varphi_i, r_1), \label{tau-assumption}
\end{equation}
then for any resource $r_0\in {\cal R}$ we have
\begin{align*}
S_{c,q}, r_0\models {\sf A}[\varphi_1 {\sf U}\varphi_2]\ \text{iff}\ q\vdash \tau_{\sf U}({\sf A}[\varphi_1 {\sf U}\varphi_2], r_0, \emptyset).
\end{align*}
\label{lemma:au}
\end{lemma}

\begin{proof}
Assume (A1).
Given a set $X\subseteq {\cal R}$ of resources, we say that $G(S_{c,q}, r_0, \varphi_2)$ is {\em $X\text{-disjoint}$} if no node of $G(S_{c,q}, r_0, \varphi_2)$ is contained in $X$. To avoid clutter, we will write $G[r_0]$ for $G(S_{c,q}, r_0, \varphi_2)$. We prove that for any set~$X\subseteq {\cal R}\setminus\{r_0\}$ of resources, 
\[
\begin{array}{ccc}
G[r_0]\ \text{fulfills}\ {\sf A}[\varphi_1 {\sf U}\varphi_2] & {\it iff} & q\vdash \tau_{\sf U}({\sf A}[\varphi_1 {\sf U}\varphi_2], r_0, X).\\
\text{and}\ G[r_0]\ \text{is}\ X\text{-disjoint}
\end{array}
\]
The proof proceeds by descending induction on the size of the set~$X$. Note that for the case $X = \emptyset$ we have $G[r_0]$ fulfills ${\sf A}[\varphi_1 {\sf U}\varphi_2]$ iff 
$q\vdash \tau_{\sf U}({\sf A}[\varphi_1 {\sf U}\varphi_2], r_0, \emptyset)$. This case proves the lemma because $G[r_0]$ fulfills ${\sf A}[\varphi_1 {\sf U}\varphi_2]$ iff $S_{c,q}, r_0\models {\sf A}[\varphi_1 {\sf U}\varphi_2]$.

Before we start, we expand $\tau_{\sf U}({\sf A}[\varphi_1 {\sf U}\varphi_2], r_0, X)$ to
{\small
\begin{align}
\tau(\varphi_2, r_0)\vee & \label{tau-u-var2} \\
& \Big( \tau(\varphi_1, r_0)\label{tau-u-var1}  \\
& \wedge \big( \forall{r_1\in E(r_0)\cap X}.\ \neg C_{r_0, r_1} \big) \label{tau-u-no-loop-back}\\
& \wedge \big( \forall{r_1\in E(r_0)\setminus X}. (C_{r_0, r_1} \Rightarrow \notag \\
& \hspace{60pt} \tau_{\sf U}({\sf A}[\varphi_1 {\sf U}\varphi_2], r_1, X\cup \{r_0\})) \big) \Big), \label{tau-u-succ}
\end{align}}
as defined in Figure~\ref{fig:encode}. To avoid clutter, we write, e.g., {\em (\ref{tau-u-no-loop-back}) is true} for $q\vdash \forall{r_1\in E(r_0)\cap X}.\ \neg C_{r_0, r_1}$.

\para{Base Case}
For the base case we have $X = {\cal R}\setminus\{r_0\}$. 

\begin{itemize}
\item[$\Rightarrow$:] Assume $G[r_0]$ fulfills ${\sf A}[\varphi_1 {\sf U}\varphi_2]$ and $G[r_0]$ is ${\cal R}\setminus\{r_0\}$-disjoint. From ${\cal R}\setminus\{r_0\}$, $G[r_0]$ consists of a single node, $r_0$. 
Furthermore, $r_0$ is a frontier node, and since $G[r_0]$ fulfills ${\sf A}[\varphi_1 {\sf U}\varphi_2]$, we have $S_{c,q}, r_0\models \varphi_2$. From~(\ref{tau-assumption}), we get $q\vdash \tau(\varphi_2, r_0)$. Since (\ref{tau-u-var2}) is true, we get $q\vdash \tau_{\sf U}({\sf A}[\varphi_1 {\sf U}\varphi_2], r_0, {\cal R}\setminus \{r_0\})$.
\item[$\Leftarrow$:] Assume $q\vdash \tau_{\sf U}({\sf A}[\varphi_1 {\sf U}\varphi_2], r_0, {\cal R}\setminus \{r_0\})$. 
Since the resource structure $S_{c,q}$ is deadlock-free, there is a resource $r_1$ in $E(r_0)\cap ({\cal R}\setminus \{r_0\}$) such that $q\vdash C_{r_0, r_1}$. 
It follows that (\ref{tau-u-no-loop-back}) is false. Therefore, it must be that
$q\vdash \tau(\varphi_2, r_0)$. By~(\ref{tau-assumption}), $S_{c,q}, r_0\models\varphi_2$. By definition of the eventuality graph $G[r_0]$, we conclude that it consists of a single node, $r_0$. It is immediate that $G[r_0]$ fulfills $ {\sf A}[\varphi_1 {\sf U}\varphi_2]$ and that it is ${\cal R}\setminus \{r_0\}$-disjoint.
\end{itemize}

\para{Inductive Step}
Assume that for any set $X\subseteq {\cal R}\setminus \{r_0\}$ of size $k\leq |X|< |{\cal R}|$, 
$G[r_0]$ fulfills ${\sf A}[\varphi_1 {\sf U}\varphi_2]$ and 
it is $X$-disjoint
iff
$\tau_{\sf U}({\sf A}[\varphi_1 {\sf U}\varphi_2], r_0, X)$.
We show that this holds for any set $X\subset {\cal R}\setminus\{r_0\}$ with $|X| = k-1$.

%
%
\begin{itemize}

\item[$\Rightarrow$:]
Assume $G[r_0]$ fulfills ${\sf A}[\varphi_1 {\sf U}\varphi_2]$ and it is $X$-disjoint.

\begin{itemize}
\item[Case 1:] 
If $S_{c,q}, r_0\models \varphi_2$, then from (\ref{tau-assumption}) we get $q\vdash \tau(\varphi_2, r_0)$. Since~(\ref{tau-u-var2}) is true, it is immediate that $q\vdash \tau_{\sf U}({\sf A}[\varphi_1 {\sf U}\varphi_2], r_0, X)$.

\item[Case 2:] If $S, r_0\not\models \varphi_2$, then by (\ref{tau-assumption}) we have $q\not\vdash \tau(\varphi_2, r_0)$.
Therefore, (\ref{tau-u-var2}) is false, so we need to show that (\ref{tau-u-var1}), (\ref{tau-u-no-loop-back}), 
and (\ref{tau-u-succ}) are all true:
\begin{compactitem}
\item
Since $G[r_0]$ fulfills ${\sf A}[\varphi_1 {\sf U}\varphi_2]$, we have
$S, r_0\models \varphi_1$ because $r_0$ is an interior node. By~(\ref{tau-assumption}), we get $q\vdash \tau(\varphi_1, r_0)$, and thus (\ref{tau-u-var1}) is true.
\item If $G[r_0]$ has an edge $(r_0, r_1)$, then it must be that the resource structure $S$ has an edge $(r_0, r_1)$ and $q\vdash c((r_0, r_1))$; otherwise, the edge $(r_0, r_1)$ is removed from $S_{c,q}$. Furthermore, since $C = \{c\}$, $C$ does not contain any control variables, and so $C_{r_0, r_1} = c((r_0, r_1))$. 
Now, since $G[r_0]$ is $X$-disjoint, we know that $r_0$ does not have any successors contained in $X$. Therefore, for any successor $r_1$ of $r_0$, we have $q\not\vdash C_{r_0, r_1}$. We conclude that (\ref{tau-u-no-loop-back}) is true.
%
%
\item Finally, consider an edge $r_1\in E(r_0)\setminus X$ such that $q\vdash C_{r_0, r_1}$.
Since $G[r_0]$ fulfills ${\sf A}[\varphi_1 {\sf U}\varphi_2]$ and $r_1$ is a successor of $r_0$, it follows that
$G(S_{c,q}, r_1, \varphi_2)$ also fulfills ${\sf A}[\varphi_1 {\sf U}\varphi_2]$. 
Furthermore, since $G[r_0]$ is $X$-disjoint, $G(S_{c,q}, r_1, \varphi_2)$ must be also $X$-disjoint. Furthermore, $G(S_{c,q}, r_1, \varphi_2)$ does not contain the node $r_0$ because $G[r_0]$ is acyclic. We conclude that $G(S_{c,q}, r_1, \varphi_2)$ is $X\cup \{r_0\}$-disjoint and it fulfills ${\sf A}[\varphi_1 {\sf U}\varphi_2]$. By induction, we get $q\vdash \tau_{\sf U}({\sf A}[\varphi_1 {\sf U}\varphi_2], r_1, X\cup \{r_0\})$. Therefore, (\ref{tau-u-succ}) is true.
\end{compactitem}

\end{itemize}

\item[$\Leftarrow$:]
Assume $q\vdash \tau_{\sf U}({\sf A}[\varphi_1 {\sf U}\varphi_2], r_0, X)$.
\begin{itemize}

\item[Case 1:] If~$q\vdash \tau(\varphi_2, r_0)$, then (\ref{tau-u-var2}) is true. By~(\ref{tau-assumption}), $S_{c,q}, r_0\models \varphi_2$. It is immediate that 
$G[r_0]$ consists of a single node, namely $r_0$. Therefore, $G[r_0]$ fulfills ${\sf A}[\varphi_1 {\sf U}\varphi_2]$ and it is $X$-disjoint because $X\subset {\cal R}\setminus \{r_0\}$.

\item[Case 2:] If~$q\not\vdash \tau(\varphi_2, r_0)$, then (\ref{tau-u-var2}) is false. Therefore, 
(\ref{tau-u-var1}), (\ref{tau-u-no-loop-back}), 
and (\ref{tau-u-succ}) must be true. 
%
%
From (\ref{tau-u-var1}) and~(\ref{tau-assumption}), we have $S_{c,q}, r_0\models \varphi_1$. 
Consider any node $r_1\in E(r_0)\setminus X$ such that $q\vdash C_{r_0, r_1}$. Then, $r_1$ is a successor of $r_0$ in the graph $G[r_0]$.
From~(\ref{tau-u-succ}), we get $q\vdash \tau_{\sf U}({\sf A}[\varphi_1 {\sf U}\varphi_2], r_1, X\cup \{r_0\})$. By induction, $G(S_{c,q}, r_1, \varphi_2)$ fulfills ${\sf A}[\varphi_1 {\sf U}\varphi_2]$ and it is $X\cup \{r_0\}$-disjoint.
Since $r_0$ is an internal node, $S, r_0\models \varphi_1$, and all subgraphs rooted
at $r_0$'s successors fulfill ${\sf A}[\varphi_1 {\sf U}\varphi_2]$, it follows that 
$G[r_0]$ fulfills ${\sf A}[\varphi_1 {\sf U}\varphi_2]$.
Furthermore, from~(\ref{tau-u-no-loop-back}) we know that $r_0$ has no successors in $X$.
Since all subgraphs rooted at $r_0$'s successors are $X\cup \{r_0\}$-disjoint, it follows that $G[r_0]$ is $X$-disjoint.
\end{itemize}

\end{itemize}

This concludes our proof.
\end{proof}

\begin{lemma}
Let $S = ({\cal R}, E, r, S)$ be a resource structure, $\varphi_1$ and $\varphi_2$ be two access constraints, and $C = \{c\}$ be a configuration template. 
If for any resource $r_1\in {\cal R}$ and any access request $q\in {\cal Q}$ we have
\begin{equation}
\tag{A2}
\bigwedge_{i\in\{1,2\}} S_{c,q}, r_1\models \varphi_i\ \text{iff}\ q\vdash \tau(\varphi_i, r_1), \label{tau_eu_assumption}
\end{equation}
then for any resource $r_0\in {\cal R}$ we have
\begin{align*}
S_{c,q}, r_0 \models {\sf E}[\varphi_1 {\sf U}\varphi_2]\ \text{iff}\ q\vdash \tau({\sf E}[\varphi_1 {\sf U}\varphi_2], r_0, \emptyset).
\end{align*}

\label{lemma:eu}

\end{lemma}

\begin{proof}
Given a directed graph $G = ({\cal R}, E)$ and a subset $X\subset {\cal R}$ of resources, we 
define {\em the projection of $G$ on $X$} as $G|_X = (X, \{ (r_0, r_1)\in E\mid \{r_0, r_1\}\subseteq X\})$. We will write $G[r_0]$ for $G(S_{c,q}, r_0, \varphi_2)$.
We prove by induction on the size of the set $X$ that for any $\{r_0\}\subseteq X\subseteq {\cal R}$, we have
\[
G[r_0]|_X\ \text{fulfills}\ {\sf E}[\varphi_1 {\sf U}\varphi_2]\ {\it iff}\
q\vdash \tau_{\sf U}({\sf E}[\varphi_1 {\sf U}\varphi_2], r_0, {\cal R}\setminus X).
\]
Note that since $G|_{\cal R} = G$ and $\tau_{\sf U}({\sf E}[\varphi_1 {\sf U}\varphi_2], r_0, {\cal R}\setminus {\cal R}) = \tau_{\sf U}({\sf E}[\varphi_1 {\sf U}\varphi_2], r_0, \emptyset)$, the case for $X = {\cal R}$ proves the lemma.

We first expand $\tau_{\sf U}({\sf E}[\varphi_1 {\sf U} \varphi_2], r_0, {\cal R}\setminus X)$ to
\begin{align}
\tau(\varphi_2, r_0) \vee & \label{tau_eu_var2} \\
& \Big( \tau(\varphi_1, r_0) \label{tau_eu_var1}\\
& \wedge \exists{r_1\in E(r_0)\!\setminus\! ({\cal R}\!\setminus\! X)}.\ \big( C_{r_0, r_1} \notag \\
& \hspace{20pt} \wedge \tau_{\sf U}({\sf E}[\varphi_1 {\sf U} \varphi_2], r_1, ({\cal R}\setminus X) \cup \{r_0\}) \big) \Big) \label{tau_eu_succ}
\end{align}

\para{Base Case}
For the base case, we have $X = \{r_0\}$.
\begin{itemize}

\item[$\Rightarrow$:] Assume that $G(S_{c,q}, r_0, \varphi_2)|_{\{r_0\}}$ fulfills ${\sf E}[\varphi_1 {\sf U}\varphi_2]$. The graph $G[r_0]|_{\{r_0\}}$ consists of the single node $r_0$. The node $r_0$ is a frontier node, and therefore it must be that $S_{c,q}, r_0\models \varphi_2$. By~(\ref{tau_eu_assumption}), we have $q\vdash \tau(\varphi_2, r_0)$. Then (\ref{tau_eu_var2}) is true and therefore $q\vdash \tau_{\sf U}({\sf E}[\varphi_1 {\sf U}\varphi_2], r_0, {\cal R}\setminus \{r_0\})$.

\item[$\Leftarrow$:] Assume that $q\vdash \tau_{\sf U}({\sf E}[\varphi_1 {\sf U}\varphi_2], r_0, {\cal R}\setminus \{r_0\})$. Since here $X = \{r_0\}$ and $S$'s edge relation is irreflexive, we have $E(r_0)\setminus ({\cal R}\setminus \{r_0\}) = \emptyset$. Therefore, (\ref{tau_eu_succ}) is false and it must be that $q\vdash \tau(\varphi_2, r_0)$. By~(\ref{tau_eu_assumption}), we have $S_{c,q}, r_0\models \varphi_2$. It is immediate that the graph $G[r_0]|_{\{r_0\}}$ consists of the single node $r_0$, and that it fulfills ${\sf E}[\varphi_1 {\sf U}\varphi_2]$.

\end{itemize}

\para{Inductive Step}
Assume that $G[r_0]|_X$ fulfills ${\sf E}[\varphi_1 {\sf U}\varphi_2]$ iff
$q\vdash \tau_{\sf U}({\sf E}[\varphi_1 {\sf U}\varphi_2], r_0, {\cal R}\setminus X)$ holds for any set $\{r_0\}\subseteq X\subset {\cal R}$ of size $1 \leq |X| \leq k$, for some $k$, $1\leq k < |R|$. We show that this also holds for any set $\{r_0\}\subset X\subseteq {\cal R}$ of size $|X| = k+1$.

\begin{itemize}

\item[$\Rightarrow$:] Assume $G[r_0]|_X$ fulfills ${\sf E}[\varphi_1 {\sf U}\varphi_2]$. 

\begin{itemize}

\item[Case 1:] If $S_{c,q}, r_0\models \varphi_2$, then from (\ref{tau_eu_assumption}) we get $q\vdash \tau(\varphi_2, r_0)$. It is immediate that $q\vdash \tau_{\sf U}({\sf E}[\varphi_1 {\sf U}\varphi_2], r_0, {\cal R}\setminus X)$.

\item[Case 2:] If $S_{c,q}, r_0\not\models \varphi_2$, then from (\ref{tau_eu_assumption}) we get $q\not\vdash \tau(\varphi_2, r_0)$. Since $G[r_0]|_X$ fulfills ${\sf E}[\varphi_1 {\sf U}\varphi_2]$, $r_0$ is an internal node and $S_{c,q}, r_0\models \varphi_1$. By (\ref{tau_eu_assumption}), $q\vdash \tau(\varphi_1, r_0)$, and so (\ref{tau_eu_var1}) is true. Furthermore, $r_0$ has a successor $r_1$ with $q\vdash C_{r_0, r_1}$ such that $r_1$ has a path to a node $r_n$ with $S_{c,q}, r_n\models \varphi_2$.
We conclude that $G(S_{c,q}, r_1, \varphi_2)|_{(X \setminus \{r_0\})}$ fulfills ${\sf E}[\varphi_1 {\sf U}\varphi_2]$. 
By induction, since $|X\setminus \{r_0\}| = k-1$, we have $q\vdash \tau_{\sf U}({\sf E}[\varphi_1 {\sf U}\varphi_2], r_1, {\cal R}\setminus (X\setminus \{r_0\}))$. 
Since $r_0\in X$, from ${\cal R}\setminus (X\setminus \{r_0\}) = ({\cal R}\setminus X)\cup \{r_0\}$ we conclude that
(\ref{tau_eu_succ}) is also true.
We conclude that $q\vdash \tau_{\sf U}({\sf E}[\varphi_1 {\sf U}\varphi_2], r_0, X)$.

\end{itemize}

\item[$\Leftarrow$:] Assume $q\vdash \tau_{\sf U}({\sf E}[\varphi_1 {\sf U}\varphi_2], r_0, X)$.

\begin{itemize}

\item[Case 1:] If $q\vdash \tau(\varphi_2, r_0)$, then from (\ref{tau_eu_assumption}) we get $S_{c,q}, r_0\models \varphi_2$. Therefore the graph $G[r_0]|_X$ consists of the single node $r_0$ with $S_{c,q}, r_0\models \varphi_2$. It is immediate that $G[r_0]|_X$ fulfills ${\sf E}[\varphi_1 {\sf U}\varphi_2]$.

\item[Case 2:] If $q\not\vdash \tau(\varphi_2, r_0)$, then it must be that (\ref{tau_eu_var1}) and (\ref{tau_eu_succ}) are true. From (\ref{tau_eu_var1}) and (\ref{tau_eu_assumption}), we get $S_{c,q}, r_0\models\varphi_1$. From (\ref{tau_eu_succ}), it follows that $r_0$ has a successor $r_1$ with $q\vdash C_{r_0, r_1}$ such that $q\vdash \tau_{\sf U}({\sf E}[\varphi_1 {\sf U}\varphi_2], r_1, ({\cal R}\setminus X)\cup \{r_0\})$. 
Since $r_0\in X$, we have $({\cal R}\setminus X)\cup \{r_0\} = {\cal R}\setminus (X\setminus \{r_0\})$.
By induction, $G(S_{c,q}, r_1, \varphi_2)|_{X\setminus \{r_0\}}$ fulfills ${\sf E}[\varphi_1 {\sf U}\varphi_2]$, so there is a path from $r_1, \ldots, r_n$ in $G(S_{c,q}, r_1, \varphi_2)|_{X\cup \{r_0\}}$ along nodes in $X\setminus \{r_0\}$ where $S_{c,q}, r_n\models\varphi_2$ and $S_{c,q}, r_i\models \varphi_1$ for $1\leq i< n$. It is immediate that there is a path $r_0, r_1, \ldots, r_n$ in $G[r_0]|_X$ such that $S_{c,q}, r_n\models\varphi_2$ and $S_{c,q}, r_i\models \varphi_1$ for $1\leq i< n$. Therefore $G[r_0]|_X$ fulfills ${\sf E}[\varphi_1 {\sf U}\varphi_2]$.

\end{itemize}

\end{itemize}

This concludes our proof.
\end{proof}


\begin{lemma}
Given a resource structure~$S = ({\cal R}, E, r, L)$, a set $R = \{T_1\Rightarrow \varphi_1, \ldots, T_n\Rightarrow \varphi_n \}$ of requirements, and a configuration template $C = \{c\}$, let $\phi = \textsc{Encode}(S, T_1\Rightarrow \varphi_1, C) \wedge \cdots \wedge \textsc{Encode}(S, T_n\Rightarrow \varphi_n, C)$. The constraint $\forall {a}.\ \phi$ is satisfiable iff $S, c\Vdash R$.
\label{lemma:smt-correct-reqs}
\end{lemma}

\begin{proof}
Note that since $C = \{c\}$, the formula $\phi$ contains no control variables, i.e.\ it contains only attribute variables. 
Since there is a one-to-one mapping from a valuation of the attribute variables $\vec{a}$ to an access request~$q$, we have $\forall \vec{a}.\ \phi$ iff $\forall q\in {\cal Q}.\ q\vdash \phi$. 
We expand the constraint $\phi$ to $(T_1\Rightarrow \tau(\varphi_1, r)) \wedge \cdots \wedge (T_n\Rightarrow \tau(\varphi_n, r))$. We get that $\forall \vec{a}.\ \phi$ iff for any access request $q\in {\cal Q}$, and for any requirement $T\Rightarrow \varphi$, $q\vdash T$ implies $q\vdash \tau(\varphi, r)$. 
By Lemma~\ref{lemma:smt-encoding}, $q\vdash \tau(\varphi, r)$ iff $S_{c,q}, r\models \varphi$. We get $\forall \vec{a}.\ \phi$ iff 
for any access request $q\in {\cal Q}$, and for any requirement $T\Rightarrow \varphi$, $q\vdash T$ implies $S_{c,q}, r\models \varphi$. By definition of $\Vdash$, we get  $\forall \vec{a}.\ \phi$ iff $S, c\Vdash R$. 
\end{proof}

We now restate and prove Theorem~\ref{thm:soundness}.

\noindent {\bf Theorem~\ref{thm:soundness}.}
{\em Let $S$ be resource structure, $R$ a set of requirements, and $C$ a configuration template. If ${\cal S}_{\sf smt}(S, R, C) = c$ then $S, c\Vdash R$. If ${\cal S}_{\sf smt}(S, R, C) = {\sf unsat}$, then there is no configuration $c$ in $C$ such that $ S, c\Vdash R$.}

\begin{proof}
Let $C = \{c_1, \ldots, c_n\}$.  The formula $\exists{\vec{z}}.
\forall{\vec{a}}.\ \phi$ generated by Algorithm~\ref{alg:smt} is
equivalent to the formula $(\forall{\vec{a}}.\ \phi_{c_1})\vee
\cdots\vee (\forall{\vec{a}}.\ \phi_{c_n})$ where $\phi_{c_i}$ is the
formula obtained by grounding the control variables $\vec{z}$ in
$\phi$ with those values that encode the configuration $c_i$. Note
that each formula $\phi_{c_i}$ is equivalent to the one obtained when
using a configuration template $C_i = \{c_i\}$.

Assume that $\exists c\in C.\ S,c\Vdash R$. By Lemma~\ref{lemma:smt-correct-reqs},  $\forall{\vec{a}}.\ \phi_{c_i}$ is satisfiable for some $c_i$ in $C$. 
Therefore ${\cal S}_{\sf smt}(S, R, C)$ returns some configuration $c_i$. Assuming the \textsc{Derive} procedure correctly derives a configuration $c_i$ from a model of $\exists{\vec{z}}\forall{\vec{a}}.\ \phi$, then 
${\cal S}_{\sf smt}(S, R, C) = c_i$  for some $c_i$ such that $\forall{\vec{a}}.\ \phi_{c_i}$. By Lemma~\ref{lemma:smt-correct-reqs}, $S, c_i\Vdash R$.

Assume that $\neg\exists c\in C.\ S,c\Vdash R$. By Lemma~\ref{lemma:smt-correct-reqs}, $\forall{\vec{a}}.\ \phi_{c_i}$ is not satisfiable for any $c_i$ in $C$. Therefore ${\cal S}_{\sf smt}(S, R, C)$ returns ${\sf unsat}$.

\end{proof}
\ 
\fi

\end{document}